\documentclass[11pt]{article}

\usepackage{geometry}
\usepackage{amssymb}
\usepackage{url}
\usepackage{amsmath}
\usepackage[square,numbers]{natbib}
\usepackage{enumerate}
\usepackage{float}
\usepackage{tikz}
\usetikzlibrary{backgrounds}
\usepackage[colorlinks=true]{hyperref}
\hypersetup{
	linkcolor=[rgb]{0.4,0,0.6},
	citecolor=[rgb]{0, 0.4, 0},
	urlcolor=[rgb]{0.6, 0, 0}
}

\usepackage[ruled]{algorithm2e}

\SetAlFnt{\small}
\SetAlCapFnt{\small}
\SetAlCapNameFnt{\small}
\SetAlCapHSkip{0pt}
\IncMargin{-\parindent}

\newtheorem{theorem}{Theorem}[section]

\newtheorem{example}{Example}[section]
\newtheorem{corollary}[theorem]{Corollary}
\newtheorem{lemma}[theorem]{Lemma}
\newtheorem{proposition}[theorem]{Proposition}

\newtheorem{claim}[theorem]{Claim}
\newtheorem{definition}[theorem]{Definition}

\newtheorem{observation}[theorem]{Observation}

\def\squarebox#1{\hbox to #1{\hfill\vbox to #1{\vfill}}}
\newcommand{\qed}{\hspace*{\fill}
\vbox{\hrule\hbox{\vrule\squarebox{.667em}\vrule}\hrule}\smallskip}
\newenvironment{proof}{\noindent{\bf Proof:~~}}{\(\qed\)}

\newcommand{\Rev}{\operatorname{Rev}}
\newcommand{\VCG}{\operatorname{VCG}}
\newcommand{\VCGUD}{\operatorname{VCG^{UD}}}
\newcommand{\Mye}{\operatorname{Mye}}

\newcommand{\OPT}{\operatorname{OPT}}

\newcommand{\vals}{\mathbf{v}}

\newcommand{\E}{\mathop{\mathbb{E}}}
\newcommand{\singleparam}{\mathrm{SingleParam}}

\newcommand{\Feas}{\mathcal{I}}
\newcommand{\Add}{\operatorname{Add}}
\newcommand{\DC}{\operatorname{DC}}
\newcommand{\Asym}{\operatorname{Asym}}
\newcommand{\Mat}{\operatorname{Mat}}
\newcommand{\UD}{\operatorname{UD}}

\newcommand{\val}{v}
\newcommand{\reals}{\mathbb{R}}

\definecolor{MyGray}{rgb}{0.8,0.8,0.8}

\newcommand{\squishlist}{
	\begin{list}{$\bullet$}
		{ \setlength{\itemsep}{0pt}      \setlength{\parsep}{3pt}
			\setlength{\topsep}{2pt}       \setlength{\partopsep}{0pt}
			\setlength{\leftmargin}{1.5em} \setlength{\labelwidth}{1em}
			\setlength{\labelsep}{0.5em} } }
	
	\newcommand{\squishend}{
\end{list}  }

\begin{document}

\title{The Competition Complexity of Auctions: \\A Bulow-Klemperer Result for Multi-Dimensional Bidders}

\author{Alon Eden\thanks{Computer Science, Tel-Aviv University. \texttt{alonarden@gmail.com}. Work done in part while the author was visiting the Simons Institute for the Theory of Computing.} \and Michal Feldman\thanks{Computer Science, Tel-Aviv University, and Microsoft Research. \texttt{michal.feldman@cs.tau.ac.il}.} \and Ophir Friedler\thanks{Computer Science, Tel-Aviv University. \texttt{ophirfriedler@gmail.com}.} \and Inbal Talgam-Cohen\thanks{Computer Science, Hebrew University of Jerusalem. \texttt{inbaltalgam@gmail.com}.} \and S.~Matthew Weinberg\thanks{Computer Science, Princeton University. \texttt{smweinberg@princeton.edu}. Work done in part while the author was a research fellow at the Simons Institute for the Theory of Computing.}}


\maketitle

\thispagestyle{empty}\maketitle\setcounter{page}{0}

\begin{abstract}

A seminal result of Bulow and Klemperer [1989] demonstrates the power of competition for extracting revenue: when selling a single item to $n$ bidders whose values are drawn i.i.d. from a regular distribution, the simple welfare-maximizing VCG mechanism (in this case, a second price-auction)
with one additional bidder extracts at least as much revenue in expectation as the optimal mechanism.
The beauty of this theorem stems from the fact that VCG is a {\em prior-independent} mechanism, where the seller possesses no information about the distribution, and yet, by recruiting one additional bidder it performs better than any prior-dependent mechanism tailored exactly to the distribution at hand (without the additional bidder).



In this work, we establish the first {\em full Bulow-Klemperer} results in {\em multi-dimensional} environments, proving that by recruiting additional bidders, the revenue of the VCG mechanism surpasses that of the optimal (possibly randomized, Bayesian incentive compatible) mechanism. For a given environment with i.i.d. bidders, we term the number of additional bidders needed to achieve this guarantee the environment's {\em competition complexity}.

Using the recent duality-based framework of Cai et al. [2016] for reasoning about optimal revenue, we show that the competition complexity of $n$ bidders with additive valuations over $m$ independent, regular items is at most $n+2m-2$ and at least $\log(m)$. We extend our results to bidders with additive valuations subject to downward-closed constraints, showing that these significantly more general valuations increase the competition complexity by at most an additive $m-1$ factor. We further improve this bound for the special case of matroid constraints, and provide additional extensions as well.

\end{abstract}

\newpage

\section{Introduction} \label{sec:intro}
A great deal of research in recent years has been devoted to the design of simple mechanisms, motivated in part by their desirability in practical settings~\cite[see][and many subsequent works]{CHK07, HR-09, HartN12, HN13}. {Different measures of mechanism complexity have been considered, including computation and communication complexity, as well as the simplicity of the mechanism's description, and its dependence on details of the economic environment.}
This last measure is motivated by early work  that has come to be known as {\em Wilson's doctrine} \cite{Wil87}, by which simple detail-free mechanisms should be preferred in complex settings in order to alleviate the risks introduced by various assumptions. Indeed, Wilson considered the ``progress of game theory'' to depend on freeness from such assumptions.  

This idea is reflected in the distinction between {\em
prior-independent}  and {\em prior-dependent} mechanisms. In both
cases it is assumed that there exist populations (prior
distributions) from which the bidders are drawn. The difference is
that prior dependent mechanisms are custom-made to these priors, the
details of which are assumed to be fully known. By contrast, the
designer of a prior-independent mechanism assumes no knowledge of
these priors. In terms of practical applicability, prior-dependent
mechanisms risk overfitting to particular beliefs, whereas
prior-independent mechanisms are inherently robust.

Of course, if we could get exactly the same guarantees from prior-independent mechanisms as prior-dependent ones (as is the case for welfare maximization by the VCG mechanism~\cite{vic-61, cla-71, gro-73}), we would simply demand that all mechanisms be prior-independent. But this is provably not the case for revenue optimization, where the best prior-dependent mechanisms strictly outperform the best prior-independent ones 
(as is the case for Myerson's mechanism~\cite{Mye81}). So the game becomes to understand the tradeoff between simplicity and optimality. 
One approach in this direction that has achieved remarkable success in the past decade is to apply the lens of approximation, and address questions of the form: ``what fraction of the optimal revenue can be achieved by a `simple' mechanism?''~\cite{Seg03, CHK07, HR-09, CHMS10, CMS10, DHKN-11, RTY15, KleinbergW12, HartN12, AzarKW13, LiY13, BILW14, DRY15, RubinsteinW15, BateniDHS15, Yao15, ChawlaM16, GoldnerK16}.\footnote{{Some of these works study prior independence ~\cite{Seg03, HR-09, DHKN-11, RTY15, AzarKW13, DRY15, GoldnerK16}, and others propose mechanisms that are inherently simple and robust, such as ``sell each item separately.''}}

While these works certainly offer an explanation for the ubiquity of simple mechanisms in practice, they don't quite tell the whole story. The lens of approximation views the environment as fixed, and tries to find a suitable auction that generates as much revenue as possible. 
But in some settings, particularly in high-stakes markets where small constant fractions amount to large losses, constant-factor approximation guarantees -- 
even with matching impossibility results -- may be unsatisfactory. However, once we enter the realm of practically-motivated questions, there's no reason to view the environment as fixed -- the seller could, for instance, spend additional effort recruiting extra bidders to attend the auction instead of treating the number of participants as given. Indeed, a typical economist might (reasonably) demand the optimal revenue and settle for nothing less, and instead pose the question as seeking the minimum cost change to the environment so that a simple mechanism achieves this guarantee.


A beautiful result in this spirit is the Bulow-Klemperer (BK) theorem~\cite{BK96}, which asserts that in the sale of a single item to $n$ symmetric bidders, whose valuations are drawn from the same (regular) distribution, running the (simple and prior-independent) VCG mechanism {(second-price auction in this case)} with $n+1$ bidders extracts at least as much revenue in expectation as Myerson's optimal (prior-dependent) mechanism with $n$ bidders. The beauty of this theorem stems from the fact that VCG is a simple and standard mechanism, completely oblivious to the bidder distribution, yet by recruiting a single additional bidder, its revenue surpasses that of the optimal auction.
%
%
The original BK result was made possible by the work of Myerson~\cite{Mye81}, thanks to which revenue optimization is extremely well-understood in single-dimensional settings. In contrast, revenue optimization is extremely poorly-understood in multi-dimensional settings, even with additive bidders (at least until recently), owing to complicating factors such as randomization~\cite{thanassoulis2004haggling,pavlov2011property,BriestCKW10, HN13}, non-monotonicity~\cite{hart2015maximal}, computational intractability~\cite{daskalakis2013mechanism, ChenDPSY14, ChenDOPSY15}, and other factors which do not arise in single-dimensional settings~\cite{RochetC98,daskalakis2014complexity,daskalakis2015strong, GiannakopoulosK15, giannakopoulos2014duality}.

\paragraph{BK results and competition complexity.}
In this paper, we establish the first {\em full BK results} for a
wide range of multi-dimensional settings. By \emph{full BK results},
we mean theorem statements of the form \emph{``the revenue of
mechanism $M$ with $n+C$ symmetric~bidders drawn from any population
{$F$} is at least as large in expectation as the optimal
revenue with $n$ bidders, as long as {$F$} satisfies condition
$X$",{\footnote{Note that $X$ encompasses both the class of
valuations of the bidders, and typically also independence and
regularity conditions on their distribution.}} } where: (a) $M$ is
simple and prior-independent -- and in this paper we further
restrict to $M$ = VCG as in the original BK result, (b) the revenue
of $M$ is required to \emph{surpass} the optimal revenue without
additional bidders (no approximation loss whatsoever), and
(c) the optimal revenue refers to the optimum among all, possibly randomized, Bayesian incentive compatible (BIC) and Bayesian individually rational (BIR)\footnote{BIC means that it is a Bayes-Nash equilibrium for all bidders to bid their true value. BIR means that agents' expected utility in this equilibrium is non-negative.} mechanisms.\footnote{This is in contrast to the work of~\cite{RTY15} which uses a strictly lower benchmark, namely, the optimum among all deterministic DSIC mechanisms. More on this in Section~\ref{sec:related}.} 

For the minimum $C$ such that a full BK statement holds for an environment (number of bidders $n$, number of items $m$, class $X$ of valuations and distributions), we say that the environment has \emph{competition complexity} $C$. 
The competition complexity thus measures how much competition is needed before the revenue of a simple mechanism reaches a strong revenue benchmark.
Competition complexity complements in a sense the measure of \emph{sample complexity} \cite{CR14,0002MR15,MohriM14,RoughgardenS16, MorgensternR16, MorgensternR15, Devanur0P16}, which measures how much data is needed before the revenue of a (\emph{non-simple}) mechanism \emph{approximates} a benchmark.



\subsection{Our Results and Techniques}
\label{sec:results}

{Before we state our results}, we remark that it is not even
a priori clear that {\textit{any}} number 
of
additional bidders suffices to reach the optimal revenue (rather
than just {approximate} 
it). In fact, without an independence
{assumption} 
on the bidders' distributions, already there exist distributions over additive valuations for two items for which the revenue of
VCG with any (finite) number of
bidders never surpasses the revenue of the optimal mechanism for just one bidder (this result is a direct corollary of
\cite{HN13}\footnote{When bidders' values are correlated,
\cite{HN13} shows there can be an infinite gap between the revenue
of the optimal randomized mechanism and the optimal deterministic
one, even with one bidder and two items. This implies there cannot
be a finite bound on the competition complexity since the revenue of
VCG from each item is bounded by the optimal revenue for each item,
which is a concave function of the number of bidders by revenue
submodularity \cite{DughmiRS12}.}). 

We consider the following distributions over valuation functions $v$: For every item $j$, $v_j$ is independently drawn from a distribution $F_j$, and for every set of items $S$, $v(S) = \max_{T \subseteq S, T \in \mathcal{I}} \{\sum_{j \in T} v_j\}$, for some downward-closed set system $([m],\mathcal{I})$. Such valuations are called ``additive subject to downward-closed constraints $\mathcal{I}$, with independent items.'' When all $F_j$ are regular, we say ``independent, regular items''. For example, when $\mathcal{I} = 2^{[m]}$, $v$ is just an additive valuation function. We also consider extensions where the constraints $\mathcal{I}$ are themselves randomly drawn (independently of item values), and in such instances replace ``downward-closed constraints $\mathcal{I}$'' with ``randomly drawn downward-closed constraints''.


\vspace{0.1in}

\noindent {\bf Theorem}: We obtain the following competition complexity results:
\squishlist 
    \item The competition complexity of $n$ bidders with additive valuations over $m$ independent, regular items is at most $n+2(m-1)$ and at least $\Omega(\log m)$.

    \item The competition complexity of $n$ bidders with additive valuations subject to downward-closed constraints $\mathcal{I}$ {over $m$ independent, regular items} is greater than the competition complexity of $n$ additive bidders over $m$ independent, regular items by at most $m-1$. There exist constraints for which this is tight.\footnote{{In fact, regularity can be replaced here (and in the next item) with any property $X$ of one-dimensional distributions.}}
    \item
    For all matroid set systems $([m],\mathcal{I})$, the competition complexity of $n$ bidders with additive valuations subject to matroid constraints
    $\mathcal{I}$ over $m$ independent, regular
    items is greater than the competition complexity of $n$ additive bidders over $m$ independent, regular items by at most $\rho$, where $\rho$ is the maximum number of disjoint independent sets in $\mathcal{I}$ that span (but do not include) item $j$, over all items $j$. There exist constraints for which this is tight.
    \item The competition complexity of $n$ bidders with additive valuations subject to {randomly drawn downwards-closed constraints} over $m$ independent, regular items is greater than the competition complexity of $n$ additive bidders over $m$ independent, regular items by at most $2(m-1)$.
\squishend 
\smallskip

{A direct implication of the above theorem is that the competition complexity of $n$ bidders with additive valuations subject to downward-closed constraints over $m$ independent, regular items is at most $n+3(m-1)$ and at least $m-1$. Note that for $k$-uniform matroids, $\rho=m-k$ , and so $\rho=m-1$ for unit-demand valuations (and $\rho=0$ for additive valuations).} The final bullet accommodates markets where, for instance, bidders are equally likely to be additive or unit-demand, independently of their values. 



\paragraph{Our techniques.}

At a high level, our approach breaks down into three steps:
(1) Find a suitable upper bound on the optimal revenue for many additive bidders, using the duality framework of~\cite{CDW16}. (2) Prove that, for additive bidders, VCG with additional bidders surpasses this benchmark, using coupling arguments and probabilistic tools. (3) Prove a ``reduction'' from arbitrary downward-closed set systems 
to additive bidders, at the cost of $m-1$ additional bidders.

\paragraph{Step One: Revenue Upper Bound.}
While we make use of the duality framework developed in~\cite{CDW16}, the upper bounds they derive for additive bidders provably don't suffice for our goals. Specifically, for a single additive bidder and any $C > 0$, there exists a distribution that is additive over two independent, regular items where the revenue of VCG with $C+1$ bidders does not exceed their bound (the example appears in Section \ref{sec:single-bidder-iid}). So our first step is to provide a new bound (still using their framework) that is suitable for our setting. Without getting into much detail, our bound can be viewed as taking a similar approach to theirs, but in ``quantile space'' as opposed to ``value space.'' That is, whenever their approach compares two random variables by their \emph{value}, we instead compare them by their \emph{quantile}.\footnote{The quantile of a random variable $X$ drawn from distribution with CDF $F$ is $F(X)$.} We postpone further details to Section~\ref{sec:new-bound}, but note here that we do expect our approach to be useful for future BK results, and multi-dimensional mechanism design in general.

\paragraph{Step Two: Covering with Extra Bidders.}
After finding a suitable bound, we need to show that the revenue of VCG with additional bidders surpasses this bound. The key tools here are coupling arguments between random variables representing the revenue of VCG with additional bidders and the random variable from Step One.

%

We wish to highlight the unique challenge in targeting a full BK result as opposed to an approximation, and postpone the details of the argument to Section~\ref{sec:multiple-additive}. Whereas in single-dimensional settings, a closed formula for the optimal revenue is well-understood, the optimal revenue in multi-dimensional settings is the expectation of an extremely bizarre random variable. Any tractable bound on the optimal revenue is therefore necessarily loose, meaning that one needs to be extra careful so as to avoid sacrificing constant-factors. Moreover, all known bounds (including ours, and others from Cai et al.'s framework) are still not ``obviously'' simple to analyze. To simplify the analysis at the cost of constant-factors, one can upper bound maximums by sums, indicator random variables by 1, etc. For instance, the now-standard ``core-tail'' decomposition follows from a sequence of operations like this. But since we want a full BK result, these simplifications are too lossy, 
necessarily complicating the analysis.

As a (surprisingly relevant) example, consider two random variables $X$ and $Y$ with the same expectation, and imagine that we wish to design an auction that gets revenue at least $\mathbb{E}[\max\{X, Y\}]$. If we are willing to lose a factor of two, we can simply design an auction that gets revenue at least $\mathbb{E}[X] = \mathbb{E}[Y]$ and be done. But if we really refuse to lose constant factors, we're forced to design a better auction.\footnote{Incidentally, we'll encounter this exact obstacle with $X$ = the ``highest virtual value'' and $Y$ = the ``second-highest value.'' Both $\mathbb{E}[X]$ and $\mathbb{E}[Y]$ capture the expected revenue of a second-price auction, but it's not obvious how one should design an auction achieving revenue $\mathbb{E}[\max\{X, Y\}]$. See Section~\ref{sec:multiple-additive}.}

\paragraph{Step Three: Reduction from arbitrary 
    downward-closed feasibility constraints to additive bidders.}
While we believe our result for additive bidders is already exciting in its own right, we are able to extend it to far more general settings by reduction. More specifically, {we first show that the revenue of the optimal mechanism for additive bidders upper bounds the optimal revenue for bidders who are additive subject to any downward-closed constraints (assuming item values are drawn from the same distributions).} From here, we show that the revenue of VCG with $\eta+m-1$ bidders who are additive subject to downward-closed constraints is at least as large as the revenue of VCG with $\eta$ additive bidders (assuming item values are drawn from the same distributions). Chaining these inequalities together, we conclude that the competition complexity increases by at most $m-1$ relative to additive bidders. 
Note that a similar result of~\cite{RTY15} for unit-demand bidders can be naturally extended (with some work) to additive bidders subject to matroid constraints, but an entirely new combinatorial approach is necessary for arbitrary downward-closed constraints, 
and this result may be of independent interest. Moreover, this reduction now allows future research to focus exclusively on improved BK results for additive bidders, as they will immediately extend far more generally.


\paragraph{Computational remarks.}
It is well known that running VCG for the most general environments where our results hold is NP-hard (see, e.g., \cite{welsh2010matroid}). Still, we wish to note the following:
\squishlist
\item Even our result on additive bidders is the first full BK result beyond single-dimensional settings, and VCG can be implemented in poly-time in this setting, in fact whenever the feasibility constraints form a matroid. 
\item Even in environments where it is NP-hard to implement VCG, competition complexity is a meaningful measure in its own right.
\item The settings that motivate full BK results versus approximations may be exactly those where practitioners find a way to solve NP-hard problems -- high stakes instances where anything less than the optimum is unsatisfactory.
\squishend

The above bullets explain why it is important to understand the competition complexity of $\VCG$, even when it is computationally intractable. Still, we provide additional competition complexity results in Appendix~\ref{app:mir_vcg} for a computationally efficient maximal-in-range (MIR) mechanism. Specifically, we show that the MIR mechanism that optimizes over the space of matchings (which can be done in poly-time) also witnesses that the competition complexity of $n$ bidders with additive valuations subject to downward-closed constraints over $m$ independent, regular items is at most $n + 3(m-1)$. 



\subsection{Related Work}
\label{sec:related}
\paragraph{Other BK and Prior-Independent Results.}
The most related work to ours is~\cite{RTY15}, which proves the only previous multi-dimensional BK result. Their result holds only for unit-demand bidders (again over independent, regular items), and compares to the optimal deterministic, dominant strategy incentive compatible (DSIC) mechanism,\footnote{A mechanism is DSIC if it is in ever bidder's interest to tell the truth, no matter what the other bidders report. It is known that BIC mechanisms can strictly outperform the best randomized, DSIC mechanism in multi-dimensional settings~\cite{Yao16}.} where 
bounds on the optimal revenue are known due to~\cite{CHMS10}. Due to the 
weaker benchmark, they are able to bound the competition complexity as $m$, which is tight.

Variants of BK results in single-dimensional settings are established in \cite{HR-09,DughmiRS12}.
There is now a rich literature on prior-independent mechanisms in both single- and multi-dimensional settings~\cite[]{Seg03, DHKN-11, azar2013optimal, AzarKW13, DRY15, GoldnerK16}. Recently, a growing literature is also considering mechanism design with limited samples from the population (whereas prior-independent mechanisms get zero samples)~\cite[e.g.,][]{CR14, Devanur0P16}. In comparison to the present paper, all of these results either study single-dimensional settings, or incur some approximation loss.

\paragraph{Related Techniques.} As already discussed, we use the recent duality framework of~\cite{CDW16} to develop our revenue upper bounds. In concurrent work, the same framework is utilized to prove approximation results for simple mechanisms in settings with many ``XOS bidders over independent items''~\cite{CaiZ16}, a single bidder with values that exhibit ``limited complementarity'' (by the present authors)~\cite{eden2016asimple}, and dynamic auctions~\cite{LiuP16}.
So step one of our approach bears similarity to the comparable steps in these works,
but each work addresses unrelated problems with techniques specific to the problem at hand. It is also worth noting that~\cite{CDW16} ``merely'' provides a \emph{framework} for deriving upper bounds on the attainable revenue, and not \emph{the} upper bound to be used always. Indeed, judging by recent applications of their framework, even selecting the appropriate bound for a given setting seems to be a bit of an art itself (although the approach across different works is of course similar). 

\paragraph{Related Themes in TCS.}
BK results resemble in spirit ideas from ``beyond worst-case complexity'', and in particular the concept of resource augmentation, where
the idea is to compare the performance of an algorithm that is endowed with more resources to an optimal outcome in an environment with less resources.
Examples of such results (which are also referred to as bicriteria results) appear in online paging \cite{SleatorT85}, network routing games \cite{roughgarden2002bad}, truthful job scheduling~\cite{chawla2013prior}, and buffer management \cite{barman2012bicriteria}.

\subsection{Organization.}
\label{sec:organization}
Preliminaries appear in Section \ref{sec:prelim}. Section \ref{sec:single-bidder-iid} is a warm-up section in which we address a single additive bidder and symmetric items (whose values are drawn i.i.d.~from a single distribution). This simple case presents a key idea, which guides us to develop a new duality-based upper-bound 
in Section \ref{sec:new-bound}, using the framework of \cite{CDW16}. In Section \ref{sec:multiple-additive} we present our main results for multiple additive bidders, and Section \ref{sec:add-st-constraints} generalizes our results to additive subject to feasibility constraints, with additional generalizations presented in Sections \ref{sub:matroids} and \ref{sub:asymmetric} for matroid and asymmetric feasibility constraints. Section \ref{sec:open} lists open research directions.

\section{Preliminaries} \label{sec:prelim}
{
There are $m$ heterogeneous items and $n$ bidders.
Every item $j$ is associated with a distribution $F_j$ with support $T_j$.
Each bidder $i$'s value for item $j$, $v_{ij}$, is drawn independently from distribution $F_j$.
Write $T = \times_{j\in [m]} T_j$ and $F= \times_{j \in [m]} F_j$.
Let $f_j(y) = \Pr_{x \gets F_j}[y = x]$ be the density of $F_j$ and write
$f(v)=\times_{j\in [m]} f_j(v_j)$.
Let $\varphi_j(\cdot)$ denote Myerson's virtual value for $F_j$.\footnote{For continuous distributions, $\varphi_j(v)=v-\frac{1-F_i(v)}{f_i(v)}$.}
In this work we assume that $F_j$ is {\em regular} for every $j$, i.e., $\varphi_j(\cdot)$ is monotone.
We consider bidders that are {\em additive subject to feasibility constraints}.
That is, each bidder $i$ has a downward-closed feasibility constraint set system $([m],\mathcal{I}_i)$, where $\mathcal{I}_i\subseteq 2^{[m]}$, and bidder $i$'s value for a set $S$ (given $\{v_{ij}\}$ and $\mathcal{I}_i$) is
$$v_i(S) = \max_{T\subseteq S, T\in \mathcal{I}_i} \{\sum_{j\in T} v_{ij} \}.$$
If $\mathcal{I}_i = \mathcal{I}$ for every bidder $i$, then we say that bidders are symmetric with respect to their feasibility constraints.
If $\mathcal{I}_i = 2^{[m]}$, then bidder $i$'s valuation is {\em additive}.
We sometimes refer to additive bidders as \textit{unconstrained} bidders.
}

Let $\vals=(v_1,\ldots,v_n)$ denote the value profile of the $n$ agents, sampled from $F^n$.\footnote{For a distribution $G$, let $G^\ell$ denote $\times_{j\in [\ell]}G$.}
Let $\vals_j=(v_{1j},\ldots, v_{nj})$ be the values of the agents for item $j$.
Let $v_{i,-j}$ denote the values sampled for agent $i$ without item $j$. Given a vector of real values $\mathbf{a}$, let $a_{(i)}$ denote the $i$-th highest value in $\mathbf{a}$.
Let $\varphi_j(v_j)^+ = \max\{\varphi_j(v_j),0\}$. Let $\mathbb{I}_A$ denote an indicator random variable for an event $A$.

\paragraph{Auction Design.}
A mechanism is a pair of an allocation and payment functions. Both are functions of the submitted bids.
We use $p_i(\cdot)$ to denote bider $i$'s (possibly randomized) payment function.
As standard, we study mechanisms that are Bayesian Individually Rational (BIR), i.e., mechanisms in which the expected utility of an agent is non-negative,
and Bayesian Incentive Compatible (BIC), i.e., mechanisms where truth-telling maximizes an agent's expected utility (in both cases, the expectation is over the randomness of the mechanism and other agents' {valuations and} strategies).

{Given a value distribution $F$, denote by $\Rev$ the expected revenue of the optimal BIR-BIC mechanism; i.e., the mechanism that maximizes $\mathbb{E}_{v\sim F}[\sum_{i\in [n]}[p_i(v)]]$ over all BIR-BIC mechanisms.}

When a single item is for sale,
the {\em second price auction with lazy reserves} first sets a reserve prices $r_i$ for each bidder $i$, and then solicits a bid from each bidder. It then allocates the item to the highest bidder $i$ if its bid surpasses $r_i$, in which case the payment is the maximum over the second highest bid and $r_i$ (if $v_i < r_i$, then the item remains unallocated).

\paragraph{Competition Complexity.} Let $\OPT (F, n)$ denote the revenue of the optimal BIR-BIC mechanism for $n$ bidders with valuations drawn i.i.d. from $F$. Let also $\Rev (M, F, k)$ denote the expected revenue of mechanism $M$ when played by $k$ bidders with valuations drawn i.i.d. from $F$. Then when we use the term ``competition complexity,'' we formally mean:

\begin{definition}\label{def:cc}
The \emph{competition complexity} $C$ with respect to a prior-independent mechanism $M$ of $n$ bidders and a class of distributions $\mathcal{D}$ is the minimum $c$ such that for all $F \in \mathcal{D}$, $\OPT (F, n) \leq \Rev (M, F, n+c)$. When $M = \VCG$, we drop ``with respect to $M$".
\end{definition}

\begin{definition}\label{def:val}[Classes of distributions considered.]
\begin{itemize}
\item \textbf{Additive over independent items}: A distribution $F$ is additive over $m$ independent items if there exist one-dimensional distributions $F_1,\ldots, F_m$, and drawing $\{v(\cdot)\} \leftarrow F$ amounts to drawing $\vec{v}=(v_1,\ldots, v_m) \leftarrow \times_j F_j$ and setting $v(S) = \sum_{j \in S} v_j$. 
\item \textbf{Subject to constraints $\mathcal{I}$}: A distribution $F$ is additive subject to constraints $\mathcal{I}$ over $m$ independent items if there exist one dimensional distributions $F_1,\ldots, F_m$, and drawing $v(\cdot) \leftarrow F$ amounts to drawing $\vec{v} \leftarrow \times_j F_j$ and setting $v(S) = \max_{T \subseteq S, T \in \mathcal{I}}\{\sum_{j \in T} v_j\}$. 
\item \textbf{Subject to randomly drawn constraints}: A distribution $F$ is additive subject to randomly drawn constraints over $m$ independent items if there exists a distribution $F_0$ over $2^{2^{[m]}}$ (set systems over $[m]$), and one-dimensional distributions $F_1,\ldots, F_m$, and drawing $v(\cdot) \leftarrow F$ amounts to drawing $\mathcal{I} \leftarrow F_0$, $\vec{v} \leftarrow \times_j F_j$ {(independently)}, and setting $v(S) = \max_{T \subseteq S, T \in \mathcal{I}}\{\sum_{j \in T} v_j\}$.
\end{itemize}
When all referenced $F_i$ are regular, we say ``independent, regular items''. 
\end{definition}

\paragraph{Classic Bulow-Klemperer.}

\begin{theorem}[\cite{BK96}, in the language of competition complexity] The competition complexity of $n$ bidders with valuations for a single, regular item is $1$.
\end{theorem}

\paragraph{Discrete vs. Continuous Distributions.} Like~\cite{CDW16}, we only explicitly consider distributions with finite support. Like their results, all of our results immediately extend to continuous distributions as well via a discretization argument of~\cite{daskalakis2012symmetries,RubinsteinW15,hartline2010bayesian,hartline2011bayesian,bei2011bayesian}. We refer the reader to~\cite{CDW16} for the formal statement and proof.

\paragraph{Tie Breaking.} Throughout the paper, we assume that there are no ties in the values of the agents. If there are ties, we break them lexicographically (first by agent's i.d., then by item's i.d.). With this tie breaking, we can also assume that there exists a unique welfare maximizing allocation of {items to agents} 
 --- given two allocations with the same welfare, we break ties in the symmetric difference according to the agent's i.d. and then the item's i.d.

\paragraph{Duality Benchmarks.} Equation~(\ref{eq:bench0}) presents the upper bound on $\Rev$ from \cite{CDW16} when there is a single bidder (for details the reader is referred to Appendix~\ref{app:duality}. In particular, Eq. \eqref{eq:bench_single_cdw}).
Let $A_j$ be the event that $\forall j'\ne j: v_j>v_{j'}$ and let $\bar A_j$ be the event that $\exists j' : v_j< v_{j'}$, then:
\begin{align}
\Rev \leq
\sum_j \sum_{v_j}f_j(v_j) \left( \varphi_j(v_j)^+\cdot \Pr_{v_{-j}}[A_j] + v_j\cdot \Pr_{v_{-j}}[\bar A_j] \right)\label{eq:bench0}
\end{align}


\section{Warm-up: Single Bidder} \label{sec:single-bidder}
In this section, we give an upper bound on  the competition complexity of a single additive bidder.
These results will convey intuition regarding the techniques used in proving BK results in more complex  settings, and for the benchmark we will use in order  to prove our  results.
\subsection{Symmetric Items} \label{sec:single-bidder-iid}
Consider a single additive bidder, and identically distributed items, i.e., $F_j=F_{j'}$ for all $j$ and $j'$, and therefore,
$\varphi_j =\varphi$ for all $j$.
Let $\sum_j \Mye_m(j)$ be the expected revenue from selling each item separately to $m$ single-parameter i.i.d. bidders  using Myerson's optimal mechanism.
We write a sum over items even though they are identically distributed because in Section~\ref{sec:new-bound} we
describe how the proof changes when items are not necessarily identically distributed.
\begin{lemma}	\label{pro:work-horse}
	$\sum_j \Mye_m(j) \ge \Rev$.
\end{lemma}

\begin{proof}
	Consider the following mechanism $M$ for selling item $j$ to $m$ single-parameter i.i.d. bidders: Run the second price auction with a personal lazy reserve price of $\varphi^{-1}(0)$ {\em only} for bidder $j$ (and a lazy reserve of $0$ for all other bidders). 
	Let $u$ denote the value profile of the $m$ bidders.
To analyze the expected revenue of this mechanism we consider the revenue from bidder $j$ --- a random variable denoted by $p_j(u)$ --- separately from the revenue from the other bidders, denoted by $p_{-j}(u)$. By the optimality of Myerson's mechanism, it holds that
	$$
	\Mye_m(j)\ge \mathbb{E}_u[p_j(u)]+\mathbb{E}_u[p_{-j}(u)].
	$$
	
	We couple $u$ and $v$ by simply setting $u=v$. Consider first $p_j(u)$. By Myerson's theorem, the expected revenue from bidder $j$ equals his expected virtual surplus \cite{Mye81}.
	By our coupling, and since the items are identically distributed, bidder $j$ is allocated in {Mechanism $M$} with precisely the same probability that event $A_j$ occurs and the bidder's virtual surplus is non-negative. 
	We get that $\mathbb{E}_u[p_j(u)] = \sum_{v_j}f_j(v_j) \cdot \varphi(v_j)^+\cdot \Pr_{v_{-j}}[A_j].$
	Consider next $p_{-j}(u)$. Using the fact that in the second price auction bidders other than $j$ win whenever event $\bar A_j$ occurs and pay at least $v_j$, we have $\mathbb{E}_u[p_{-j}(u)] \ge \sum_{v_j}f_j(v_j)\cdot v_j\cdot \Pr_{v_{-j}}[\bar A_j].$
	Comparing the bounds on $p_j(u)$ and $p_{-j}(u)$ summed over all items to Equation \eqref{eq:bench0} completes the proof.
\end{proof}
\begin{corollary}
	The competition complexity of a single bidder whose valuation is additive over $m$ i.i.d. regular items is at most $m$.
\end{corollary}
\begin{proof}
	Consider the revenue of $\VCG$ with additional $m$ bidders. By additivity, $\VCG$ is separable over the items, and therefore, the revenue is exactly the sum of second price auctions run for each item separately with $m+1$ bidders.
	By applying the classic Bulow-Klemperer theorem on each item separately, we get that $\VCG$ with $m+1$ bidders gives revenue greater than $\sum_j \Rev_m(j) \ge\Rev$, where the last inequality follows from Lemma \ref{pro:work-horse}.
\end{proof}

We note that this proof does not extend when the items are not identically distributed. The reason for this is that now the probability that bidder $j$ wins in the single parameter environment is not identical to the probability that event $A_j$ happens.
In fact, next we present  
an example where the previous benchmark used in \cite{CDW16} fails to provide any meaningful bounds on the competition complexity when items' values are drawn from different distributions.

\subsection{Example Where the Previous Benchmark Fails}
\begin{proposition}\label{prop:counterexample}
	For every $C>0$, there exists a value distribution $F$ that is additive over $2$ (non-i.i.d.) independent, regular items, 
	 where $\Rev (\VCG, F, C+1)$ is less than the benchmark of Cai et al. \cite{CDW16} for one bidder with valuations drawn from $F$ (Equation~\eqref{eq:bench0}).
\end{proposition}

\begin{proof}
	Consider a market with a single bidder and only two items. Item $a$ is distributed according to the equal revenue distribution capped at some value $k$; \textit{i.e.}, $$F_a(x)=\begin{cases}
	1-1/x \quad & x < k\\
	1 \quad & x=k
	\end{cases},$$
	supported on $[1,k]$.
	Item $b$ is distributed according to the following CDF: $F_b(x)=1-\frac{1}{x-k+1},$ supported on $[k,\infty)$. One can verify that both distributions are indeed regular.

	Consider the benchmark \eqref{eq:bench0} proposed by \cite{CDW16}, which in the case of two items is interpreted as the expected virtual value of the item with the highest value of the two \textit{plus} the expected value of the item with the lower value of the two. Since item $b$ always has a higher value than item $a$, this is exactly the expected virtual value of item $b$ \textit{plus} the expected value of item $a$. The expected value of item $a$ is 

	\begin{eqnarray*}
		\int_{0}^{k}(1-F(x)) dx + k\cdot \frac{1}{k} = 1 + \int_{1}^{k}\frac{1}{x} dx + 1 = 2+\ln k.
	\end{eqnarray*}
	As for the expected virtual value of item $b$, 
	 one can readily compute that $\mathbb{E}[\varphi_b(x)] = k$.\footnote{One easy way to shortcut this computation is the apply Myerson's Theorem (revenue = virtual welfare) and immediately conclude that $\mathbb{E}[\varphi_b(x)] = k \cdot (1-F(k)) = k$.} Combining the expected value of item $a$ and the expected virtual value of item $b$, we get that the benchmark is at least $2+\ln k + k$.
	
	Now consider running $\VCG$ using $\ell$ bidders in total ($\ell-1$ additional bidders). Since running $\VCG$ with additive bidders is the same as running $\VCG$ for each item separately, we analyze how much revenue we get by running $\VCG$ for each item. We first analyze the revenue from running Myerson's optimal mechanism on item $a$ with $\ell$ bidders (which is an upper bound on running $\VCG$). The expected revenue of the optimal single-bidder auction is $1$.
	By revenue submodularity (\cite{DughmiRS12}), the expected revenue of running the optimal mechanism is a concave function of the number of bidders. Thus, the expected revenue of the optimal auction with $\ell$ bidders is at most $\ell$.
	
	\begin{figure}[t]
		\centering
		\includegraphics[scale=0.65]{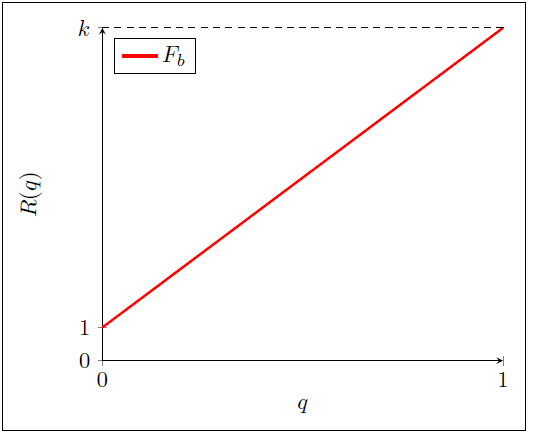}
		\caption{The revenue curve of distribution $F_b$ used in the example.}
		\label{fig:rev_curve_Gb}
				
	\end{figure}
%
%
%
%
%
%
%
	
	We now compute the revenue of the optimal mechanism for item $b$ with $\ell$ bidders. For a single bidder, the optimal mechanism sets a take-it-or-leave-it price of $k$ on the item, and gets an expected revenue of $k$. For two bidders, the optimal mechanism is $\VCG$ with monopoly reserves. Since the monopoly reserve is $k$, this is equivalent to just running $\VCG$. We compute the expected revenue of $\VCG$ the following way --- fix bidder $1$, and compute her expected payment. Bidder $2$'s value sets a random threshold for bidder $1$. If the quantile of bidder $2$'s value is $q$, then the bidder $1$'s expected payment is exactly the probability her value is above $v(q)$ \textit{times} $v(q)$. This is exactly $R(q)$. Therefore, the expected payment of bidder $1$ is exactly the area under the revenue curve, which is equal to $(k+1)/2$. By symmetry, bidder $2$'s expected payment is $(k+1)/2$ as well, and the expected revenue overall is $k+1$. Again, by submodularity of revenue, since the second bidder adds $1$ to the expected revenue, the expected revenue of $\VCG$ with $\ell$ bidders is at most $k+\ell-1$.
	
	We get that the expected revenue obtained by running $\VCG$ on both items with $\ell$ bidders is at most $k+2\ell-1$. Therefore, $\ell$ has to be at least $(\ln k + 3)/2 = \Omega(\ln k)$ in order to cover the benchmark. Since $k$ can be arbitrary large, this benchmark is not suitable in providing any meaningful bound on the competition complexity for this market.
\end{proof} 
\subsection{Asymmetric Items} \label{sec:new-bound}
{
In this section we present a new upper bound on the optimal revenue and demonstrate its applicability by showing that it implies a competition complexity result for the case of asymmetric items (and an additive bidder).
}


\paragraph{New Duality-Based Bound: Quantile Based Regions.}

As previously discussed, the proof of the case where the items are
sampled from the same distribution does not extend when the items
are not i.i.d. The reason is that the probability of the $j$th
sample to be the highest one  is not the same as the probability of
bidder $j$ to have the highest value for item $j$, since the samples
of the bidders' values for item $j$ {are i.i.d., while the values of the single bidder for the different items in the original scenario are not i.i.d.}
The intuition behind the
new benchmark we use is that we want to mimic the bidders' i.i.d.
samples, even if the items themselves are not i.i.d.

For each item $j$, we define item $j$'s region $R_j=\{v : \forall j'\ne j \ \  F_j(v_j) > F_{j'}(v_{j'}) \},$
that is, for each bidder, region $j$ is defined as the region where the bidder's value is of the highest quantile according to the items' distributions.
Using this definition, we give a new upper bound on the revenue.
\begin{theorem}[New upper bound on $\Rev$]\label{thm:new_bound}
    For every product distribution $F$, and $n$ additive bidders, whose valuations are sampled i.i.d. from $F$,
    \begin{eqnarray}
        \Rev \leq \sum_{j=1}^m\E_{\vals\sim F^n}\left[\max_{i\in [n]}\bigg(\varphi_{j}(v_{ij})^+\cdot\mathbb{I}_{v_i\in R_j} \ + \ v_{ij}\cdot\mathbb{I}_{v_i\notin R_j}\bigg)\right].\label{eq:multiple_new_flow}
    \end{eqnarray}
\end{theorem}
Let $B_j$ be the event that $\forall j'\ne j: F_j(v_j)\ge F_{j'}(v_{j'})$ and let $\bar B_j$ be the event that $\exists j' : F_j(v_j)<F_{j'}(v_{j'})$.
A direct corollary of Theorem \ref{thm:new_bound} gives the following on the revenue from a single bidder.
\begin{corollary} \label{cor:single_new_bound}
    For every product distribution $F$, the revenue of a single additive bidder, whose valuation is sampled from $F$, is upper bounded by
        \begin{equation}
    \sum_j \sum_{v_j}f_j(v_j) \left( \varphi_j(v_j)^+\cdot \Pr_{v_{-j}}[B_j] + v_j\cdot \Pr_{v_{-j}}[\bar B_j] \right).
    \label{eq:bench_noniid_single}
    \end{equation}
\end{corollary}
Now, one may notice that the probability of event $B_j$ equals exactly the probability that in $m$ i.i.d. samples from $F_j$, the $j$th sample is the highest. The derivation of the new upper bound, which uses the duality-based framework of \cite{CDW16}, and the resulting corollary, are deferred to Appendix \ref{app:duality}. In particular, 
{corollary~\ref{cor:single_new_bound}} stems from Eq.
\eqref{eq:bench_single_ours}.

\paragraph{Application to a Single Additive Bidder, Asymmetric Items.}
With the new bound in hand, we show how to sidestep the impossibility result
of Proposition \ref{prop:counterexample} associated with the upper
bound of~\cite{CDW16}.

\begin{theorem}
    The competition complexity of a single bidder whose valuation is additive over $m$ independent, regular items is at most $m$.
\end{theorem}

\begin{proof}
    The proof is almost identical to the proof in Section \ref{sec:single-bidder-iid}.
    We state what changes in the proof of Proposition \ref{pro:work-horse} when we use the upper bound in \eqref{eq:bench_noniid_single}.
    $u$ is now a vector of i.i.d.~draws from $F_j$, whereas $v$ is a vector of independent draws from $F_1,F_2,\dots,F_m$. We couple $u$ and $v$ so that
    $F_{j'}(v_{j'})=F_{j}(u_{j'})$
    for every coordinate $j'$.
     Still, bidder $j$ is allocated in mechanism $M$ with precisely the probability that event $B_j$ occurs and its virtual value is positive (we are using here the monotonicity of $F_j$, i.e., that $u_j>u_{j'}$ for all $j'\ne j$ iff $F_j(v_j)=F_j(u_j)>F_j(u_{j'})=F_{j'}(v_{j'})$ for all $j'\ne j$). And in the second price auction, bidders other than $j$ win whenever event $\bar{B}_j$ occurs, and pay at least $v_j$.
\end{proof}

%
\section{Main Result: Bulow-Klemperer for Multiple Additive Bidders}
\label{sec:multiple-additive}
In this section, we are given an instance with $n$ additive bidders and $m$ heterogeneous item, where bidders' valuations are drawn i.i.d. from the product of regular distributions $F$. We show that by adding $n+2m-2$ bidders and running VCG, we are able to get at least as much revenue as the optimal obtainable revenue by any BIC mechanism with the original $n$ bidders.
\begin{theorem}\label{thm:additive_main} The competition complexity of $n$ bidders with additive valuations over $m$ independent, regular items is at most $n+2m-2$, and at least $\Omega(\log m)$.
\end{theorem}

The remainder of this section contains a proof of the upper bound, with some missing proofs of claims deferred to Appendix \ref{app:mult_add_missing}. The lower bound is proved in Appendix~\ref{sec:lb}, and makes use of the same ``i.i.d. equal revenue'' example from~\cite{HartN12}.

\subsection{An Upper Bound on Item $j$'s Contribution to the Revenue}
Fix some item $j$, and consider $j$'s contribution to the upper bound on the revenue given in Eq.~\eqref{eq:multiple_new_flow}, i.e.,
$$\Rev_j = \E_{\vals\sim F^n}\left[\max_{i\in [n]}\bigg(\varphi_{j}(v_{ij})^+\cdot\mathbb{I}_{v_i\in R_j} \ + \ v_{ij}\cdot\mathbb{I}_{v_i\notin R_j}\bigg)\right].$$
For a given valuation profile $\vals$, let ``$(i)$'' be the index of the bidder with the $i$th highest valuation for item $j$. The next key lemma gives a bound on $\Rev_j$, which is useful in proving our BK result.
\begin{lemma} \label{lem:revj_bound}
	$\Rev_j$ is upper bounded by
	\begin{eqnarray*}
\E_{\textbf{v}\sim F^n}\left[\max\left\{\varphi_j\left(v_{(1)j}\right),v_{(2)j}\right\}\ \vert \  v_{(1)}\in R_j  \right]\cdot \Pr_{\textbf{v}\sim F^n}\left[ v_{(1)}\in R_j\right]
 +  \E_{\textbf{v}\sim F^n}\left[v_{(1)j} \ \vert \  v_{(1)}\notin R_j\right]\cdot \Pr_{\textbf{v}\sim F^n}\left[ v_{(1)}\notin R_j\right].
	\end{eqnarray*}
\end{lemma}
\begin{proof}
	We consider two cases. If, $v_{(1)}\notin R_j$, then $$v_{(1)j}\geq v_{ij}\geq \varphi_{j}(v_{ij})^+\cdot\mathbb{I}_{v_i\in R_j} \ + \ v_{ij}\cdot\mathbb{I}_{v_i\notin R_j}$$
	for every bidder $i$, and therefore, $$v_{(1)j}\geq \max_{i\in [n]}\bigg(\varphi_{j}(v_{ij})^+\cdot\mathbb{I}_{v_i\in R_j} \ + \ v_{ij}\cdot\mathbb{I}_{v_i\notin R_j}\bigg).$$
	On the other hand, if $v_{(1)}\in R_j$, then
	\begin{eqnarray*}
		& & \max_{i\in [n]}\bigg(\varphi_{j}(v_{ij})^+\cdot\mathbb{I}_{v_i\in R_j} \ + \ v_{ij}\cdot\mathbb{I}_{v_i\notin R_j}\bigg) = \\ \\& &\max\left\{\varphi_j\left(v_{(1)j}\right)^+, \max_{i\neq (1)}\bigg(\varphi_{j}(v_{ij})^+\cdot\mathbb{I}_{v_i\in R_j} \ + \ v_{ij}\cdot\mathbb{I}_{v_i\notin R_j}\bigg)\right\}\leq \\
		& & \max\left\{\varphi_j\left(v_{(1)j}\right)^+, v_{(2)j}\right\} = \max\left\{\varphi_j\left(v_{(1)j}\right), v_{(2)j}\right\}.
	\end{eqnarray*}
	
	We get that
	\begin{eqnarray}
	\Rev_j & \leq & \E_{\vals\sim F^n}\left[ \max\left\{\varphi_j\left(v_{(1)j}\right), v_{(2)j}\right\}\cdot \mathbb{I}_{v_{(1)}\in R_j} + v_{(1)j}\cdot \mathbb{I}_{v_{(1)}\notin R_j} \right] \nonumber\\
	& = & \E_{\textbf{v}\sim F^n}\left[\max\left\{\varphi_j\left(v_{(1)j}\right),v_{(2)j}\right\}\ \vert \  v_{(1)}\in R_j  \right]\cdot \Pr_{\textbf{v}\sim F^n}\left[ v_{(1)}\in R_j\right]\nonumber\\
	& & +  \E_{\textbf{v}\sim F^n}\left[v_{(1)j} \ \vert \  v_{(1)}\notin R_j\right]\cdot \Pr_{\textbf{v}\sim F^n}\left[ v_{(1)}\notin R_j\right]. \label{eq:j_contrib}
	\end{eqnarray}
\end{proof}

\subsection{A Single Parameter Lemma}
In this section we prove our main technical lemma.
\begin{lemma}\label{lem:main_additive}
	Consider the case with a single item for sale and $2n+2m-2$ bidders whose values are drawn i.i.d. from $F_j$.
    The expected revenue obtained by running $\VCG$ (just on item $j$) is at least $\Rev_j$.
\end{lemma}
We first show that proving this lemma immediately yields a proof for Theorem \ref{thm:additive_main}.\\

\noindent\textbf{Proof of Theorem \ref{thm:additive_main}} (based on Lemma \ref{lem:main_additive}): For additive bidders, running VCG is equivalent to running VCG on each item separately. Applying Lemma \ref{lem:main_additive}, we get that $$\Rev \leq \sum_{j=1}^m \Rev_j\leq \sum_{j=1}^{m}\E_{\vals_j\sim F_j^{2n+2m-2}}\VCG(\vals_j)= \E_{\vals\sim F^{2n+2m-2}}\VCG(\vals).$$\qed

In proving Lemma~\ref{lem:main_additive}, we make use of the following property of $\VCG$, which is folklore knowledge (proof appears in Appendix~\ref{app:mult_add_missing} for completeness):

\begin{observation}[Folklore] \label{obs:vcg_opt}
	Consider a set of bidders drawn i.i.d. from a regular distribution and a single item for sale. The optimal mechanism that always sells an item is the $\VCG$ mechanism.
\end{observation}

In order to prove Lemma $\ref{lem:main_additive}$, we introduce a mechanism for selling a single good to a set of $2n+2m-2$ bidders drawn from $F_j$, and show that the expected revenue of this mechanism is at least the bound on $\Rev_j$ in Eq. $(\ref{eq:j_contrib})$. Moreover, this mechanism always sells the item, and therefore, by Observation \ref{obs:vcg_opt}, VCG gets at least as much revenue, so a bound on the revenue of this mechanism is sufficient to prove Lemma $\ref{lem:main_additive}$. Our single-item mechanism, $\singleparam\mbox{-}j$, is given in Fig. $\ref{fig:single_param}$.

\begin{figure} [H]
	\colorbox{MyGray}{
		\begin{minipage}{\textwidth} {
		$\singleparam\mbox{-}j$\\
		\noindent\textbf{Input:} $2n+2m-2$ bids sampled i.i.d. from distribution $F_j$.
		\begin{enumerate}
			\item Consider the bidders in some arbitrary, predetermined order. Rename the bidders as follows:
			\begin{enumerate}
				\item Rename the first $n$ bids to $u_{(1)j},u_{(2)j},\ldots,u_{(n)j}$, where $u_{(i)j}$ is the $i$th highest bid of the first $n$ bids. Let $\textbf{u}_j$ denote this $n$-dimensional vector.
				\item Rename the next $m-1$ bids to $u_{(1)1},u_{(1)2},\ldots,u_{(1)j-1},u_{(1)j+1},\ldots,u_{(1)m}$ according to their arbitrary order. Let $u_{(1)-j}$ denote this $(m-1)$-dimensional vector.
				\item Rename the last $n+m-1$ bids to $w_{(1)j},w_{(2)j},\ldots,w_{(n+m-1)j}$, where $w_{(i)j}$ is the $i$th highest bid of the last $n+m-1$ bids. Let $\textbf{w}$ denote this $(n+m-1)$-dimensional vector.
			\end{enumerate}
			\item Let $j'\gets {\arg\max}_{k\neq j} u_{(1)k}$.
			\item If $u_{(1)j'}> u_{(1)j}$: \label{step:comparison1}
			\squishlist 
				\item Allocate the item to $u_{(1)j'}$.
			\squishend 
			\item Else:
			\begin{enumerate}
				\item If $\varphi_j\left(u_{(1)j}\right) > w_{(2)j}$: \label{step:comparison2}
				\squishlist 
					\item Allocate the item to $u_{(1)j}$.
				\squishend 
				\item Else:
				\squishlist 
					\item Allocate the item to $w_{(1)j}$.
				\squishend 
			\end{enumerate}
			\item Charge the winner according to Myerson's payment identity.
		\end{enumerate}
		}
	\end{minipage}}
	\caption{A single parameter mechanism with expected revenue at least $\Rev_j$.}
	\label{fig:single_param}
\end{figure}

We first note that mechanism $\singleparam\mbox{-}j$ is truthful since the allocation of every bidder is monotone in the bidder's valuation \cite{Mye81}. Before giving the full proof of Lemma~\ref{lem:main_additive}, we give the intuition. The trick is trying to couple the various events $v_{(1)}\in R_j$, $\varphi_j(v_{(1)j}) > v_{(2)j}$, etc. that affect the random variable related to $\Rev_j$ (Eq. \eqref{eq:j_contrib}) with events that determine the different cases in the mechanism $\singleparam\mbox{-}j$. The reader can get a good idea for which events will be coupled based on our decision for how to allocate the item in $\singleparam\mbox{-}j$. However, the complete analysis is a bit subtle and in particular requires the following probabilistic claim (proof appears in Appendix~\ref{app:prob_claim}) in addition to the proper coupling.

\begin{claim} \label{cor:max_fresh_geq_max_same}
	For any regular distribution $F$, $$\E_{\substack{\mathbf{a}\sim F^{\ell}\\ \mathbf{b}\sim F^{k}}}\left[\max\left\{\varphi\left(a_{(1)}\right),a_{(2)}\right\}\vert a_{(1)} > b_{(1)}\right] \leq \E_{\substack{\mathbf{a}\sim F^{\ell}\\ \mathbf{b}\sim F^{k}\\ \mathbf{c}\sim F^{\ell+k}}}\left[\max\left\{\varphi\left(a_{(1)}\right),c_{(2)}\right\}\vert a_{(1)} > b_{(1)}\right].$$	
\end{claim}

Note that in the LHS above, the random variables $a_{(1)}$ and $a_{(2)}$ are correlated, but come from $\ell$ bidders. In the RHS, $a_{(1)}$ and $c_{(2)}$ are independent, but $\textbf{c}$ comes from $\ell+k$ bidders. We proceed to prove our  main technical lemma.
\newline\newline
\noindent\textbf{Proof of Lemma \ref{lem:main_additive}:}
	Consider the values $\vals_j$ and $v_{(1)-j}$ in the upper bound of $\Rev_j$ in Eq. $(\ref{eq:j_contrib})$ and $\textbf{u}_j$ and $u_{(1)-j}$ as defined in the renaming phase of $\singleparam\mbox{-}j$. Couple the values whenever $\vals_j=\textbf{u}_j$ and $F_{j'}(v_{(1)j'})=F_j(u_{(1)j'})$ for every $j'\neq j$. To see that this is a valid coupling, observe that both $\vals_j$ and $\textbf{u}_j$ are vectors of $n$ i.i.d. samples from $F_j$, and that $F_{j'}(v_{(1)j'})$ and $F_j(u_{(1)j'})$ are sampled uniformly over the interval $[0,1]$  for all $j'$, independent from other samples.
	
	We next analyze the expected revenue obtained by $\singleparam\mbox{-}j$. 
	For the bidders in $\textbf{u}_j$, we compute their expected virtual value, and for the other bidders, we compute their expected payment. 
	We distinguish between two cases: 
	
	~
	
	\noindent {\bf Case (1):} $u_{(1)j'}> u_{(1)j}$ for some $j'$ (the condition checked in step $\ref{step:comparison1}$ of the mechanism). In this case the item goes to some bidder $u_{(1)j'}$. Since bidder $u_{(1)j'}$ pays the minimum value at which she is allocated, her payment is at least $u_{(1)j}$. 
	The probability for this case is
	\begin{eqnarray*}
		\Pr_{\substack{u_{(1)-j}\sim F_j^{m-1}\\ \textbf{u}_j\sim F_j^n}}\left[\exists j' : u_{(1)j}< u_{(1)j'} \right] & = & \Pr_{\substack{u_{(1)-j}\sim F_j^{m-1}\\ \textbf{u}_j\sim F_j^n}}\left[ \exists j' : F_j(u_{(1)j})<F_{j}(u_{(1)j'})\right]\\
		& = & \Pr_{\substack{v_{(1)-j}\sim F_{-j} \\ \textbf{v}_j\sim F_j^n}}\left[ \exists j' : F_j(v_{(1)j})<F_{j'}(v_{(1)j'})\right]\\
		& = & \Pr_{\textbf{v}\sim F^n}\left[ v_{(1)}\notin R_j\right],
	\end{eqnarray*}
	where the second equality follows from the coupling and the last equality follows from the definition of $R_j$ and the independence of the samples.
	
	~
	
	\noindent {\bf Case (2):} $u_{(1)j'} < u_{(1)j}$ for all $j'$. 
	This event happens with the complementary probability of Case (1), which is $\Pr_{\textbf{v}\sim F^n}\left[v_{(1)}\in R_j\right]$. In this case, the winner is determined according to the condition in Step~\ref{step:comparison2} of our mechanism. I.e., if $\varphi_j\left(u_{(1)i}\right) > w_{(2)j}$, then the winner is bidder $u_{(1)i}$;  
	otherwise \big(if $\varphi_j\left(u_{(1)i}\right) < w_{(2)j}$\big), the winner is $w_{(1)j}$, and by the payment identity, her payment is at least $w_{(2)j}$. 
	
	Combining the two cases, the expected revenue of $\singleparam\mbox{-}j$ is at least
	\begin{eqnarray*}
		& & \E_{\substack{u_{(1)-j}\sim F_j^{m-1}\\ \textbf{u}_j\sim F_j^n}}\left[u_{(1)j} \ \vert \  \exists j'\neq j : u_{(1)j} < u_{(1)j'} \right]\cdot \Pr_{\textbf{v}\sim F^n}\left[ v_{(1)}\notin R_j\right] \\
		& & + \E_{\substack{u_{(1)-j}\sim F_j^{m-1}\\ \textbf{u}_j\sim F_j^n\\ \textbf{w}\sim F_j^{n+m-1}}}\left[\varphi_j\left(u_{(1)j}\right)\cdot \mathbb{I}_{\varphi_j\left( u_{(1)j}\right) \geq w_{(2)j}} + w_{(2)}\cdot \mathbb{I}_{\varphi_j\left( u_{(1)j}\right) < w_{(2)j}} \ \vert \  \forall j'\neq j : u_{(1)j} > u_{(1)j'} \right]\cdot \Pr_{\textbf{v}\sim F^n}\left[ v_{(1)}\in R_j\right].
	\end{eqnarray*}
	We bound each summand of the bound above separately.
	For the first summand, we have
	\begin{eqnarray*}
		\E_{\substack{u_{(1)-j}\sim F_j^{m-1}\\ \textbf{u}_j\sim F_j^n}}\left[u_{(1)j} \ \vert \  \exists j'\neq j : u_{(1)j} < u_{(1)j'} \right]
		& = & \E_{\substack{u_{(1)-j}\sim F_j^{m-1}\\ \textbf{u}_j\sim F_j^n}}\left[u_{(1)j} \ \vert \  \exists j'\neq j : F_j\left(u_{(1)j}\right) < F_j\left(u_{(1)j'}\right) \right]\\
		& = & \E_{\substack{v_{(1)-j}\sim F_{-j}\\ \textbf{v}_j\sim F_j^n}}\left[v_{(1)j} \ \vert \  \exists j'\neq j : F_j\left(v_{(1)j}\right) < F_{j'}\left(v_{(1)j'}\right) \right]\\
		& = & \E_{\textbf{v}\sim F^n}\left[v_{(1)j} \ \vert \  v_{(1)}\notin R_j\right],
	\end{eqnarray*}
	where the second equality follows by the coupling, and the last equality follows by the definition of $R_j$ and the independence of samples.
	
	For the second summand, we first notice that $\varphi_j\left(u_{(1)j}\right)\cdot \mathbb{I}_{\varphi_j\left( u_{(1)j}\right) \geq w_{(2)j}} + w_{(2)}\cdot \mathbb{I}_{\varphi_j\left( u_{(1)j}\right) < w_{(2)j}}=\max\left\{\varphi_j\left(u_{(1)j}\right),w_{(2)}\right\}$. We have
	\begin{eqnarray*}
		& &\E_{\substack{u_{(1)-j}\sim F_j^{m-1}\\ \textbf{u}_j\sim F_j^n\\ \textbf{w}\sim F_j^{n+m-1}}}\left[\max\left\{\varphi_j\left(u_{(1)j}\right),w_{(2)}\right\}\ \vert \  \forall j'\neq j : u_{(1)j} > u_{(1)j'} \right]\\
		& \geq & \E_{\substack{u_{(1)-j}\sim F_j^{m-1}\\ \textbf{u}_j\sim F_j^n}}\left[\max\left\{\varphi_j\left(u_{(1)j}\right),u_{(2)j}\right\}\ \vert \  \forall j'\neq j : u_{(1)j} > u_{(1)j'} \right]\\
		& = & \E_{\substack{v_{(1)-j}\sim F_{-j}\\ \textbf{v}_j\sim F_j^n}}\left[\max\left\{\varphi_j\left(v_{(1)j}\right),v_{(2)j}\right\}\ \vert \  \forall j'\neq j : F_j\left(v_{(1)j}\right) > F_{j'}\left(v_{(1)j'}\right)  \right]\\
		& = & \E_{\textbf{v}\sim F^n}\left[\max\left\{\varphi_j\left(v_{(1)j}\right),v_{(2)j}\right\}\ \vert \  v_{(1)}\in R_j  \right],
	\end{eqnarray*}
	where the inequality follows Claim \ref{cor:max_fresh_geq_max_same}, the first equality follows by the coupling and the last equality follows by the definition of $R_j$ and the independence of samples.
	
	Consequently, the revenue of $\singleparam\mbox{-}j$ is at least
	\begin{eqnarray*}
		\E_{\textbf{v}\sim F^n}\left[v_{(1)j} \ \vert \  v_{(1)}\notin R_j\right]\cdot \Pr_{\textbf{v}\sim F^n}\left[ v_{(1)}\notin R_j\right]
		+  \E_{\textbf{v}\sim F^n}\left[\max\left\{\varphi_j\left(v_{(1)j}\right),v_{(2)j}\right\}\ \vert \  v_{(1)}\in R_j  \right]\cdot \Pr_{\textbf{v}\sim F^n}\left[ v_{(1)}\in R_j\right],
	\end{eqnarray*}
	which  is greater than $\Rev_j$ according to Lemma $\ref{lem:revj_bound}$. Since $\Rev_j$ is a mechanism that always sells the item, applying Observation \ref{obs:vcg_opt} completes the proof of Lemma \ref{lem:main_additive}.\qed

\section{Generalizations: Additive Subject to Downward-Closed}
\label{sec:add-st-constraints}
In this section we generalize our results beyond additive bidders, to bidders whose valuations are additive subject to downwards-closed constraints.
In Section \ref{sub:downward} we consider distributions where the constraints are fixed for the entire distribution, and downwards-closed. In Section~\ref{sec:extensions}, we state the following extensions (all proofs are in the appendix):
\begin{itemize}
\item In the special case that the constraints are matroids, we obtain improved guarantees on the competition complexity (Theorem~\ref{cor:matroid-main}, proved in Appendix~\ref{sub:matroids}).
\item In the case that the constraints themselves are part of the distribution (``additive subject to randomly drawn downward-closed constraints''), we obtain slightly weaker guarantees (Theorem~\ref{thm:ccasymmetric}, proved in Appendix~\ref{sub:asymmetric}).
\item Finally, we extend our results to the competition complexity with respect to a poly-time maximal-in-range mechanism. This complements our results with respect to $\VCG$ in settings for which $\VCG$ cannot be implemented in poly-time (Theorem~\ref{thm:vcgud}, proved in Appendix~\ref{app:mir_vcg}).
\end{itemize}

\subsection{Symmetric General Downward-Closed Constraints}
\label{sub:downward}

In this section, the bidders have valuations that are additive subject to identical downward-closed feasibility constraints, represented by the set system $([m], \Feas)$. We can assume without loss of generality that for every item $j$, $\{j\}\in \mathcal{I}$. 
Our main result in this section is the following:

\begin{theorem}[Competition complexity for downward-closed]
	\label{thm:downward-main}
	Let the competition complexity of $n$ additive bidders over $m$ independent, regular items be $X=X(n,m)$. Then the competition complexity of $n$ additive bidders subject to downward-closed constraints $\mathcal{I}$ over $m$ independent regular items is at most $X + m-1$.
\end{theorem}

Theorem \ref{thm:downward-main} relies on the Lemma \ref{lem:downward-extra} below. Corresponding lemmas for matroids and asymmetric feasibility constraints appear in Appendices \ref{sub:matroids}  and \ref{sub:asymmetric}, respectively (in particular, Lemmas \ref{lem:matriod-extra} and \ref{lem:vcg_rev_asym}). To state our main lemma we use the following notation: fix a product of $m$ regular distributions $F$; 
let $\VCG_{n}^{\Add}$ denote the expected VCG revenue from selling the $m$ items to $n$ additive bidders whose values are i.i.d.~draws from $F$; let $\VCG_{n'}^{\DC}$ denote the expected VCG revenue from selling them to $n'\geq n$ bidders with i.i.d.~values drawn from $F$, whose valuations are additive subject to $n'$ identical downward-closed feasibility constraints. Intuitively, if $n' \approx n$, the expected revenue from VCG with the $n$ \emph{unconstrained} bidders is higher than the revenue from VCG with the $n'$ \emph{constrained} bidders, since in the former all bidders compete for all items. Lemma \ref{lem:downward-extra} gives a bound on how much larger $n'$ should be relative to $n$ such that this intuition ceases to hold. 

\begin{lemma}[Main lemma for downward-closed]
	\label{lem:downward-extra}
	$\VCG_{n+m-1}^{\DC} \ge \VCG_n^{\Add}$.
\end{lemma}

The proof of Lemma \ref{lem:downward-extra} appears in Section \ref{sub:main-lemma-pf}. This lemma can be thought of as quantifying the extra competition complexity required due to the feasibility constraints of the bidders (in comparison to unconstrained additive bidders), and is the main technical hurdle in proving Theorem~\ref{thm:downward-main}. At a high level, the proof is as follows: First, observe that however VCG decides to allocate the items, it always has the option to reallocate item $j$ to a bidder who currently receives nothing. So whoever receives item $j$ pays at least the highest value for item $j$ among all bidders who receive nothing, and the trick is comparing this random variable with $n+m-1$ bidders to the second-highest value from $n$ bidders. Due to the potential complexity of arbitrary downward-closed feasibility constraints, the random variable denoting the highest value for item $j$ among bidders who receive nothing depends quite intricately on the values of all bidders for all other items, and one cannot reason about this random variable with the same Greedy-type arguments one might use for feasibility constraints that are matroids, and a more careful combinatorial argument is required instead.
Given Lemma \ref{lem:downward-extra} (proved shortly), we can prove the main theorem of this section:
\newline\newline
\noindent\textbf{Proof of Theorem \ref{thm:downward-main}:}
	Fix a product of regular distributions $F$, and a downward-closed feasibility set system. We use the following notation: let $\Rev_{n}^{\Add}$ denote the optimal expected revenue that can be achieved from selling the $m$ items to $n$ additive bidders whose values are i.i.d.~draws from $F$, and let $\Rev_{n}^{\DC}$ denote the optimal expected revenue from selling them to such bidders whose valuations are subject to the downward-closed constraint. {Note first that for every feasibility set system,
		$\Rev_{n}^{\Add} \ge \Rev_{n}^{\DC}$. This follows from the following three facts:
\squishlist 
\item Without loss of generality, the optimal mechanism for agents subject to downwards-closed constraints $\mathcal{I}$ only allocates sets in $\mathcal{I}$.
\item The designer of a mechanism for additive bidders is free to restrict herself to allocating only sets in $\mathcal{I}$.
\item Subject to this restriction, it is immaterial whether a bidder is additive or additive subject to $\mathcal{I}$.
\squishend 
So in particular, there exists an optimal mechanism for bidders constrained by $\mathcal{I}$ that is truthful for unconstrained bidders, and the optimal mechanism can only get better.
}
	
	By our assumption that the competition complexity of additive bidders is $X=X(n,m)$, we know that $\VCG_{n+X}^{\Add} \ge \Rev_{n}^{\Add}$.
	We can now apply Lemma \ref{lem:downward-extra} to get $\VCG_{n+X+m-1}^{\DC} \ge \VCG_{n+X}^{\Add}$. Putting everything together,
	$$
	\VCG_{n+X+m-1}^{\DC} \ge \VCG_{n+X}^{\Add} \ge \Rev_{n}^{\Add} \ge \Rev_{n}^{\DC}.
	$$
	We have shown that the competition complexity of additive bidders subject to downward-closed constraints is $X+m-1$, and this completes the proof of Theorem \ref{thm:downward-main}.\qed

\subsection{Proof of Main Lemma (\ref{lem:downward-extra}) for Downward-Closed}
\label{sub:main-lemma-pf}

We first introduce a lemma that gives a lower bound on the revenue of $\VCG$ with bidders subject to identical downward-closed constraints.
\begin{lemma} \label{lem:vcg_rev}
	Consider a set of bidders with identical downward-closed feasibility constraint, and the outcome of $\VCG$. For each item $j$, let $v^*_j$ be the highest value of an unallocated bidder for item $j$ (a bidder that does not receive any item). The revenue of $\VCG$ is at least $\sum_{j\in [m]} v^*_j$.
\end{lemma}
\begin{proof}
Recall that in the $\VCG$ mechanism, each bidder pays his externality on the environment; \textit{i.e.}, the difference in welfare of other bidders when he is not allocated and we he is. Let $A$ and $\bar{A}=[n]\setminus A$ be the respective sets of allocated and unallocated bidders. Fix a bidder $i\in A$, and let $S_i$ be the feasible set of items he receives in $\VCG$. Let $W$ be the welfare of the $\VCG$ outcome. The welfare of all bidders but $i$ for this allocation is $W-v_i(S_i)$.

Consider now the following allocation --- allocate the items not in $S_i$ as in the allocation of the $\VCG$ mechanism, and distribute the items in $S_i$ among the unallocated bidders, where each item $j\in S_i$ is allocated for a bidder $i'\in \bar{A}$ that has value $v^*_j$ for item $j$. Notice that since the bidders allocated items in $S_i$ receive subsets of a feasible set of items, their allocation is also feasible due to the downward-closedness of the feasibility set. The welfare of the purposed allocation is $W-v_i(S_i)+\sum_{j\in S_i}v^*_j$. Since we described some allocation of items to all bidders but $i$, the welfare of the optimal allocation to bidders without $i$ is at least as high. We get that the payment of bidder $i$ is at least $\sum_{j\in S_i} v^*_j$. Summing over all bidders completes the proof of the lemma.
\end{proof}

Our goal is to show that $\VCG_{n+m-1}^{\DC} \ge \VCG_n^{\Add}$. Recall that with additive bidders, VCG decomposes over the items. Therefore, $\VCG_n^{\Add}$ is the sum of expected revenues obtained from selling every item $j$ separately to $n$ single-parameter bidders, whose i.i.d.~values are drawn from $F_j$, using the VCG auction (in this case VCG is equivalent to the second-price auction).
To prove the lemma it is thus sufficient to show that for every item $j$, the expected revenue from selling $j$ separately to $n$ single-parameter bidders -- i.e., the expectation of the second highest among $n$ i.i.d.~random samples from $F_j$, denoted by $\E[u_{(2:n)j}]$ -- is at most the expected payment for item $j$ in VCG with $n+m-1$ constrained bidders (the contribution of selling item $j$ to $\VCG_{n+m-1}^{\DC}$).

In the remainder of the proof we shall argue that in VCG with $n+m-1$ constrained bidders, there are always at least $n$ bidders whose values for item $j$ are i.i.d.~draws from $F_j$, such that at most a single bidder from this set can be allocated. Since by Lemma \ref{lem:vcg_rev}, the payment for item $j$ is at least the highest value for $j$ among the unallocated bidders. This will show that the expected payment for $j$ is $\ge \E[u_{(2:n)j}]$, completing the proof.

From now on, fix an item $j$, and fix the values of the $n+m-1$ constrained bidders for all items but~$j$. Consider the welfare-maximizing allocation of all items but $j$ to the $n+m-1$ bidders. We denote this allocation by $\OPT_{-j}$, and the set of allocated bidders by $A$. Clearly since we have allocated $m-1$ items, $|A|\le m-1$, and $|\bar{A}|=n+m-1-|A|\ge n$ (where $\bar{A}=N\setminus A$). We notice the following:
\begin{claim} \label{clm:single_allocated}
	If item $j$ is allocated to an bidder $i\in \bar{A}$, no other bidder in $\bar{A}$ is allocated an item.
\end{claim}
\begin{proof}
	We argue that if $j$ is allocated to an bidder $i\in \bar{A}$, the allocation of all other items in $\OPT$ is identical to their allocation in $\OPT_{-j}$, and so precisely one bidder in $\bar{A}$ is allocated by VCG, as required. Items are allocated as in $\OPT_{-j}$ due to the downward-closedness of the constraints on the bidders' additive valuations; we now show that assuming otherwise leads to a contradiction: Consider the allocation $\OPT_{-j}$, and assume for contradiction that after a bidder $i\in\bar{A}$ is allocated item $j$, to get to allocation $\OPT$ we need to perform additional reallocations of items. Since we can assume there's a unique optimal allocation (as described in Section \ref{sec:prelim}), this means that we need additional reallocations to get to the optimal welfare. Because valuations are additive subject to downward-closed constraints, if we can increase the total welfare by item reallocations when bidder $i$ has item $j$, we can get the same increase in welfare by these reallocations when bidder $i$ does not have item $j$. This contradicts the optimality of $\OPT_{-j}$, and completes the proof of the claim.
\end{proof}

We now further fix the values of the bidders in $A$ for item $j$, so that all remaining randomness is in the values of the bidders in $\bar{A}$ for item $j$. The information we currently have allows us to compute an optimal allocation if in the optimal allocation $j$ is allocated to an bidder $i\in A$. We denote this allocation by $\OPT_{j\in A}$. More importantly, it lets us identify a set of bidders in $\bar{A}$ of size at least $n-1$ which will not be allocated if $j$ will be allocated to an bidder in $A$.
\begin{claim} \label{clm:unallocated}
	Before sampling the values of item $j$ for set $\bar{A}$, one can identify a set $\tilde{A}\subseteq \bar{A}$ such that $|\tilde{A}|\geq n-1$, and none of the bidders in $\tilde{A}$ will be allocated in case item $j$ is allocated to an bidder in $A$ in $\OPT$.
\end{claim}
\begin{proof}
	Whenever $j$ is allocated to some bidder $i\in A$, $OPT$ is  $\OPT_{j\in A}$. As mentioned above, $\OPT_{j\in A}$ can be computed, even though we do not know the value of item $j$ to bidders in $\bar{A}$, simply compute the welfare-maximizing allocation in which $j$ is allocated to a bidder in $A$. 
	We will now prove the existence of $n-1$ unallocated bidders from $\bar{A}$ in $\OPT_{j\in A}$. 
	
	Our proof is by the following charging argument: We show a chain of reallocations which starts from allocation $\OPT_{-j}$ and leads to allocation $\OPT_{j\in A}$, and argue that it ends with at least $n-1$ bidders in $\bar{A}$ unallocated.
	Start with bidder $i\in A$ who gets item $j$ in $\OPT_{j\in A}$. We say bidder $i$ \emph{vacates} an item $j'$ if $j'$ is allocated to $i$ in $\OPT_{-j}$ but not in $\OPT_{j\in A}$; we say bidder $i$ \emph{snatches} an item $j'$ if $j'$ is not allocated to $i$ in $\OPT_{-j}$ but $i$ adds it to his allocation. Let $i$ vacate and snatch items until $i$'s allocation reaches his allocation in $\OPT_{j\in A}$. To continue the chain we place the following bidders from $A$ into a ``queue'' (provided 
	that they are not already in the queue): those from whom $i$ snatched an item and those who end up grabbing an item that $i$ vacated. Importantly, when a bidder $i'$ enters the queue, there is a unique item to which his entrance can be \emph{charged}: either an item that was snatched from him by $i$, or an item that $i$ vacated and he grabbed.
	
	We now take the next bidder from the queue and repeat the process (letting this bidder vacate and snatch items until he reaches his allocation in $\OPT_{j\in A}$, and placing bidders in the queue as described above), until the queue is empty. Notice that by the assumption of a unique optimal allocation, 
	each bidder from $A$ not entered to the queue has the same allocation in $\OPT_{j\in A}$ and $\OPT_{-j}$. This is true since items of such an bidder are not snatched by bidders in $\bar{A}$; otherwise, this would contradict $\OPT_{-j}$'s optimality.
	
	We use this chain of reallocations to show a one-to-one mapping from bidders in $A\setminus \{i\}$ to items that are allocated to bidders in $A$ both  in $\OPT_{-j}$ and $\OPT_{j\in A}$. The mapping is as follows:
	\squishlist 
	\item Each bidder who is entered to the queue since an item of his was snatched by a different bidder in $A$ is mapped to a snatched item.
	\item Each bidder who is entered to the queue since he grabbed a vacated item is mapped to an item he grabbed.
	\item Each bidder in $A\setminus\{j\}$ not entered to the queue is mapped to an item allocated to him.
	\squishend 
	This mapping implies bidders in $A$ have at least $|A|$ items allocated to them in $\OPT_{j\in A}$ --- item $j$ is allocated to bidder $i\in A$, and all other bidders are mapped to a distinct item that is allocated to a bidder in $A$. This implies that we have at most $m-|A|$ items allocated to $\bar{A}$ in $\OPT_{j\in A}$. Since $|\bar{A}|=m+n-1-|A|$, it must hold that $n-1$ bidders from $\bar{A}$ are unallocated in $\OPT_{j\in A}$.
\end{proof}

Using the two claims above, we can now prove Lemma \ref{lem:downward-extra}.
\newline\newline
\noindent\textbf{Proof of Lemma \ref{lem:downward-extra}:}
	Let $\tilde{A}\subseteq \bar{A}$ be a set of $n-1$ bidders as guaranteed to exist by Claim \ref{clm:unallocated}. If $|\tilde{A}| < n$, take an arbitrary bidder $i'$ from $\bar{A}\setminus\tilde{A}$ and add him to $\tilde{A}$ (such an bidder is guaranteed to exist since $|\bar{A}|\geq n$ as mentioned above). We now have a set of at least $n$ bidders whose value for item $j$ is yet to be sampled. The values of these bidders for item $j$ are i.i.d.~draws from $F_j$, as required (this is known as the principle of deferred decision).
	
	To complete the proof, we now claim that at most one of the bidders from $\tilde{A}\cup \{i'\}$ is allocated in $\OPT$. We consider two cases --- If $j$ is allocated to an bidder in $\bar{A}$, then by Claim \ref{clm:single_allocated}, this is the only bidder from $\bar{A}$ who is allocated an item in $\OPT$, and therefore, the claim follows. If $j$ is allocated to and bidder in $A$, then by Claim \ref{clm:unallocated}, all bidders in $\tilde{A}$ that are different from $i'$ are unallocated in $\OPT$, and the claim follows. This implies that the payment for item $j$ in $\VCG_{n+m-1}^{\DC}$ is at least the second highest of $n$ i.i.d~draws from $F_j$, exactly the payment for item $j$  in $\VCG_{n}^{\Add}$. This completes the proof of the lemma, and therefore, of Theorem \ref{thm:downward-main}.\qed

\subsection{Extensions}\label{sec:extensions}
For the special case that the constraints are a matroid, we obtain improved bounds on the competition complexity. A proof and definition of disjoint spanning number appear in Appendix~\ref{sub:matroids}, but we note here that the disjoint spanning number is always at most $m-1$. 

\begin{theorem}[Competition complexity for matroids]
	\label{cor:matroid-main}
	Let the competition complexity of $n$ bidders with valuations that are additive over $m$ independent, regular items be $X=X(n,m)$, and let $\mathcal{I}$ be a matroid with disjoint spanning number $\rho$. Then the competition complexity of $n$ additive bidders with valuations that are additive subject to constraints $\mathcal{I}$ over $m$ independent, regular items is at most $X+\rho$.
\end{theorem}

For distributions that are additive subject to randomly drawn downward-closed constraints, we obtain the following slightly weaker guarantee on competition complexity (proof in Appendix~\ref{sub:asymmetric}).

\begin{theorem}[Competition complexity for randomly drawn downward closed constraints]\label{thm:ccasymmetric}
Let the competition complexity of $n$ bidders with additive valuations over $m$ independent, regular items be $X=X(n,m)$.
Then the competition complexity of $n$ bidders with additive valuations subject to randomly drawn downwards-closed constraints over $m$ independent, regular items is at most $X+2(m-1)$.

\end{theorem}

Finally, we extend our results to the competition complexity with respect to a poly-time maximal-in-range mechanism. Specifically, we consider the maximal-in-range mechanism that only considers allocations that are \emph{matchings} (i.e. that allocate each bidder at most one item), and denote this mechanism by $\VCGUD$. Note that $\VCGUD$ can be implemented in poly-time. A proof appears in Appendix~\ref{app:mir_vcg}.

\begin{theorem}[Competition complexity with respect to poly-time MIR mechanism]\label{thm:vcgud}
	Let the competition complexity of additive regular bidders be $X=X(n,m)$. Then the competition complexity with respect to $\VCGUD$ of $n$ bidders with additive valuations subject to randomly drawn downwards-closed constraints over $m$ independent regular items is at most $X+m-1$.	
\end{theorem}

\section{Discussion and Future Work}
\label{sec:open}
We present the first full BK results in multi-dimensional settings, and show that the competition complexity of $n$ buyers with additive valuations subject to downward-closed feasibility constraints, over independent, regular items is at most $n+3m-3$, and at least $m-1$. For the special case of additive (with no constraints), the competition complexity is at most $n+2m-2$ and at least $\Omega(\log m)$. For those who are approximation-minded in the traditional sense, an easy corollary is that for fixed $m$, as $n \rightarrow \infty$, VCG itself guarantees a 2-approximation to the optimal revenue.\footnote{This requires two easy steps, and a formal proof is in Appendix~\ref{app:large_market}.}

An obvious open question is to close the gaps in the above bounds, at least in the special case of additive or unit-demand. Even more concretely: is the right dependence on $m$ for additive bidders linear or logarithmic (or in between)? Is any dependence on $n$ necessary? The latter is especially enticing: if the answer is no, then VCG itself approaches the optimal revenue as $n \rightarrow \infty$.

Another direction is to consider other prior-independent mechanisms in order to get better guarantees. For instance, maybe the $\Omega(\log m)$ lower bound can be circumvented by considering a different prior-independent mechanism, or perhaps the better of two prior independent mechanisms, a la~\cite{BILW14}.\footnote{Of course, randomizing over the two would still lose a constant factor. The idea would be that you still need to learn something about your population in order to apply such a result, but just which of the two mechanisms is better.} Overall, while we believe our results to already be quite substantial, we strongly feel there is a wealth of open problems in the direction of competition complexity.

\bibliographystyle{acmsmall}

\appendix

\section{Duality Framework} \label{app:duality}
\paragraph{Reduced Forms.}  Reduced forms are defined for probabilities over type spaces $T=\times_{i\in [n]} T_i$, where $T_i$ is the type  space of agent $i$. The reduced form of an auction stores for
all bidders $i$, items $j$, and types $v_i$, what is the probability
that agent $i$ will receive item $j$ when reporting $v_i$ to the mechanism
(over the randomness in the mechanism and randomness
in other agents’ reported types, assuming they come from
$F_{-i}$) as $\pi_{ij}(v_i)$, as well as $p_i(v_i)$, which is the expected price bidder $i$
pays when reporting type $v_i$ (over the same randomness). 
It is easy to see that if a buyer that is additive subject to constraints,
his expected value for reporting type $v'_i$ to the mechanism is just
$v_i \cdot \pi_i(v'_i)$, so long as the mechanism only ever awards each type of each bidder a set of items that fit their constraints (i.e. whenever a type of each bidder is subject to constraints $\mathcal{I}$, the mechanism awards them a subset of items in $\mathcal{I}$). There always exists an optimal mechanism with this property, and there is furthermore a trivial transformation to turn any BIC-BIR mechanism without this property into one with this property that achieves exactly the same revenue (see, e.g.~\cite{DaskalakisW12}). We say that a reduced form is feasible if there exists
some feasible mechanism that matches the probabilities promised by
the reduced form.\footnote{For a more comprehensive description of reduced forms and their applications, see \cite{CDW16,CaiDW12,CaiDW12b,CaiDW13,CaiDW13b}.}

\begin{definition}[Reworded from~\cite{CDW16}, Definitions~2 and~3]
	\label{def:duality}
	A mapping $\lambda_i: T_i \times T_i \rightarrow \reals^+$ is \emph{flow-conserving} if for all $\val \in T_i$: $\sum_{\val '\in T_i} \lambda_i(v, v') \leq f_i(v)+ \sum_{\val ' \in T_i} \lambda_i(v', v)$.\footnote{This is equivalent to stating that there exists a $\lambda_i(v, \bot) \geq 0$ such that $\lambda_i(v, \bot)+\sum_{\val '\in T_i} \lambda_i(v, v') = f_i(v)+ \sum_{\val ' \in T_i} \lambda_i(v', v)$, which might look more similar to the wording of Definition~2 in~\cite{CDW16}.} 
	Given $\lambda_i$, the corresponding \emph{virtual value function} $\Phi_i$, is a transformation from vectors in $T_i$ to valuation vectors in $T^\times_i$ (the closure of $T_i$ under linear combinations) and satisfies:
	$$\Phi_i(\val) = \val - \frac{1}{f_i(\val)}\sum_{\val ' \in T_i} \lambda_i(\val ', \val )(\val ' - \val  ).$$
\end{definition}
Equipped with these definitions we can state the main theorem of the framework from \cite{CDW16}.
\begin{theorem}[Revenue $\leq$ Virtual Welfare \cite{CDW16}]\label{thm:revenue less than virtual welfare}
	Let $\lambda_i$ be flow conserving for all $i$, and $M = (\pi,p)$ any BIC and BIR mechanism. Then the revenue of $M$ is $\leq$  the virtual welfare w.r.t. the virtual value function $\Phi(\cdot)$ corresponding to $\lambda$. That is:
	$$\sum_{i=1}^{n} \sum_{v_i \in T_i} f_{i}(v_{i})\cdot p_i(v_i)\leq \sum_{i=1}^{n} \sum_{v_{i}\in T_{i}} f_{i}(v_{i})\cdot \pi_{i}(v_{i})\cdot\Phi_{i}(v_{i}).$$
\end{theorem}
Therefore, the task of finding a good bound reduces to finding good flow conserving duals.
In order to find flow conserving duals we begin with the following definition:
\begin{definition}[Upward-closed region]
	A {\em region} (subset) of types $R \subseteq \bigcup_i T_i$ is upward closed with respect to item $j$ if for every $v_i \in R \cap T_i$, all types $v'_i=(v'_{ij}, v_{i-j})\in T_i$ such that $v'_{ij}\geq v_{ij}$, are also in $R$. 
	We use $R_j$ to denote an upward-closed region with respect to item $j$.
\end{definition}
A method for designing flow conserving duals is the following: 
first, partition the type space $T$ into upward-closed regions. 
Then, define the duals as follows: for a type $v_i \in R_j \cap T_i$, i.e., in a region that is upward closed w.r.t. $j$.
Let $v'_i=(v'_{ij},v_{i-j})$ be the type so that $v'_{ij}$ is the largest value that is smaller than $v_{ij}$ (if exists). 
If $v'_i$ is also in $R_j$, then set $\lambda_i(v_i,v'_i) = f(v_i)+ \sum_{\hat{\val}_i\in T_i} \lambda_i(\hat{\val}_i, v_i)$, and otherwise set  
$\lambda(v_i,v'_i) = 0$.
Set all other duals to $0$.



%


Let $\Phi_i(\cdot)$ be the resulting virtual transformation.
The following lemma is essentially a restatement of Claims 1 and 2 in \cite{CDW16} to upwards closed regions.

\begin{lemma}[\cite{CDW16}]
	\label{lem:virtual} 
	Let $v_i\in T_i$ be some valuation for which $v_i\in R_j$ for some $j\in \{1,\ldots,m\}$. Then $\Phi_{ik}(v_i)=v_{ik}$ if $k\neq j$ and $\Phi_{ij}(v_i)=\varphi_{ij}(v_{ij})$ otherwise.
\end{lemma}
Using the region-based flow, and assuming each item $j$ has exactly one upward-closed region $R_j$ associated with it (as is the case throughout the entire paper), we get the following upper bound on the revenue of any BIR-BIC mechanism.
\begin{theorem}
	For any set of upward-closed regions $R_1,\ldots, R_m$, the optimal revenue a seller can obtain by a BIC and BIR mechanism is
	\begin{eqnarray}
	\Rev & \leq & \max_{\pi}\sum_{i=1}^n \sum_{v_i\in T_i}\sum_{j=1}^m f_i(v_i) \cdot \pi_{ij}(v_i) \bigg(\varphi_{ij}(v_{ij})\cdot\mathbb{I}_{v_i\in R_j} \ + \ v_{ij}\cdot\mathbb{I}_{v_i\notin R_j}\bigg) \nonumber\\
	& = & \max_{\pi}\sum_{i=1}^n \sum_{v_i\in T_i}\sum_{j=1}^m f_i(v_i) \cdot \pi_{ij}(v_i) \bigg(\varphi_{ij}(v_{ij})^+\cdot\mathbb{I}_{v_i\in R_j}\ + \ v_{ij}\cdot\mathbb{I}_{v_i\notin R_j}\bigg).\label{eq:rev_upper_bound}
	\end{eqnarray}
\end{theorem}
\begin{proof}
The inequality is a  direct corollary of Theorem \ref{thm:revenue less than virtual welfare} and Lemma \ref{lem:virtual}, and the equality follows since the max is achieved whenever $\pi_{ij}(v_i)=0$ for each $v_i\in R_j$ such that $\varphi_{ij}(v_{ij}) < 0$.
\end{proof}

\subsection{Single Additive Bidder Case}
Since there is only a single additive bidder, with no feasibility constraints, the bound in $(\ref{eq:rev_upper_bound})$ simplifies to:
\begin{eqnarray}
\Rev&\leq& \sum_{v\in T}\sum_j f(v)  \left(\varphi_j(v_{j})^+\cdot\mathbb{I}_{v\in R_j}+v_{j}\cdot \mathbb{I}_{v\notin R_j}\right)\nonumber\\
& = & \mathbb{E}_{v\sim F} \Big[ \sum_{j=1}^m\left(\ \varphi_j(v_j)^+ \cdot \mathbb{I}_{v\in R_j} + v_j \cdot \mathbb{I}_{v\notin R_j}\ \right) \Big]\nonumber\\
&= &\sum_{j=1}^m \sum_{v_j} f_j(v_j) \left( \varphi_j(v_j)^+\cdot \Pr_{v_{-j}}[v\in R_j ] + v_j\cdot \Pr_{v_{-j}}[v\notin R_j ] \right).\label{eq:bench_single_bidder}
\end{eqnarray}

In \cite{CDW16},  the regions are defined as follows: $R_j$ contains all valuations $v=(v_1,\ldots,v_m)$ for which $v_j > v_k$ for every $k \neq j$ (breaking ties lexicographically). Therefore, they get the following bound:
\begin{align} 
\Rev\leq \sum_j \sum_{v_j}f_j(v_j) \left( \varphi_j(v_j)^+\cdot \Pr_{v_{-j}}[\forall k  : v_j \geq v_k ] + v_j\cdot \Pr_{v_{-j}}[\exists k: v_j< v_k ] \right).\label{eq:bench_single_cdw}
\end{align} 

Instead, we define regions as follows: $R_j$ contains all valuations $v=(v_1,\ldots,v_m)$ for which $F_j(v_j)> F_k(v_k)$ for every $k \neq j$ (breaking ties lexicographically), and therefore, our single bidder bound has the form:
\begin{align} 
\Rev\leq \sum_j \sum_{v_j}f_j(v_j) \left( \varphi_j(v_j)^+\cdot \Pr_{v_{-j}}[\forall k  : F_j(v_j) \geq F_k(v_k) ] + v_j\cdot \Pr_{v_{-j}}[\exists k: F_j(v_j)< F_k(v_k) ] \right).\label{eq:bench_single_ours}
\end{align} 

\subsection{Multiple Additive Bidder Case}
We now present the proof of Theorem~\ref{thm:new_bound}, which is the multiple additive bidder benchmark we use to prove our main result .

\textbf{Proof of Theorem~\ref{thm:new_bound}}
Notice that in Equation~(\ref{eq:rev_upper_bound}), since there are no feasibility constraints, $\pi_{ij}(\cdot)$ can allocate each item to the non-negative,  maximizing, virtual value. Therefore:
\begin{eqnarray*}
	\Rev &\leq & \max_{\pi} \mathbb{E}_{\vals \sim F^n}\left[\sum_{i=1}^n \sum_{j=1}^m \pi_{ij}(v_i) \bigg(\varphi_{ij}(v_{ij})^+\cdot\mathbb{I}_{v_i\in R_j}\ + \ v_{ij}\cdot\mathbb{I}_{v_i\notin R_j}\bigg)\right] \\
	&= & \mathbb{E}_{\vals \sim F^n}\left[ \sum_{j=1}^m \max_{i\in [n]} \bigg(\varphi_{ij}(v_{ij})^+\cdot\mathbb{I}_{v_i\in R_j}\ + \ v_{ij}\cdot\mathbb{I}_{v_i\notin R_j}\bigg)\right] \\
	&= & \sum_{j=1}^m \mathbb{E}_{\vals \sim F^n}\left[ \max_{i\in [n]} \bigg(\varphi_{ij}(v_{ij})^+\cdot\mathbb{I}_{v_i\in R_j}\ + \ v_{ij}\cdot\mathbb{I}_{v_i\notin R_j}\bigg)\right].
\end{eqnarray*}


	
\section{Missing Proofs of Section \ref{sec:multiple-additive}}\label{app:mult_add_missing}

\noindent\textbf{Proof of Observation \ref{obs:vcg_opt}:} Since the expected revenue of a mechanism is the expected virtual value of the allocated bidder, the optimal mechanism that always sells an item allocates the item to the bidder with the highest virtual value. Since the bidders are drawn i.i.d. from a regular distribution, the bidder with the highest virtual value is also the bidder with the highest value. Thus, the optimal mechanism that always sells an item allocates the item to the bidder with the highest value; this is exactly the allocation rule of $\VCG$.\qed
 
\subsection{Proof of Claim \ref{cor:max_fresh_geq_max_same}} \label{app:prob_claim}

Before proving the claim, we introduce the following definitions and lemmas.

\begin{definition}
	Random variable $A$ first-order stochastically dominates (FOSD) random variable $B$ if for every $x$, $\Pr\left[A> x\right]\geq \Pr\left[B> x\right]$.
\end{definition}

\begin{definition}
	Consider two random variables $X$ and $Y$. We say that $X,Y$ are \emph{positively correlated} if for every $x$ and $y$, $\Pr[X>x\vert Y>y] \geq \Pr[X>x].$
	\newline(I.e., $X_{\mid Y>y}$ FOSD $X$ for every $y$.)
\end{definition}

\begin{lemma}
	\label{lem:corr}
	If $X,Y$ are positively correlated and $\widetilde{Y}$ is independent of $X$ and FOSD $Y$, then $\mathbb E[\max\{X,Y\}]\le \mathbb E[\max\{X,\widetilde Y\}]$.
\end{lemma}

\begin{proof}
	We prove the claim for $\widehat{Y}$ that has \textit{exactly} the same marginals as $Y$  but is independent of $X$. Since $\widetilde Y$ FOSD $\widehat{Y}$ and both are independent of $X$, it is immediate that $E\left[\max\{X,\widetilde{Y}\}\right]\geq E\left[\max\{X,\widehat{Y}\}\right]\geq E\left[\max\{X,Y\}\right]$ as desired.
	First, notice that
	\begin{eqnarray}
	\E\left[\max\{X,Y\}\right]& = & \int_{0}^{\infty}\Pr\left[\max\{X,Y\}>x\right] dx\nonumber\\
	& = & \int_{0}^{\infty}\Pr\left[X>x\right] + \Pr\left[Y>x\right] - \Pr\left[X>x \ \mbox{ and } \  Y > x\right] dx,\label{eq:max_cor}
	\end{eqnarray}
	and similarly
	\begin{eqnarray}
	\E\left[\max\{ X,\widehat Y\}\right] &  = &   \int_{0}^{\infty}\Pr\left[ X>x\right] + \Pr\left[\widehat Y>x\right] - \Pr\left[X>x\ \mbox{ and } \  \widehat Y > x\right] dx\nonumber\\
	& = & \int_{0}^{\infty}\Pr\left[X>x\right] + \Pr\left[Y>x\right] - \Pr\left[X>x\ \mbox{ and } \  \widehat Y > x\right] dx,\label{eq:max_ind}
	\end{eqnarray}
	where the second equality holds since $\widehat Y$ and $Y$ have the same marginals.	
	Subtracting $(\ref{eq:max_ind})$ from $(\ref{eq:max_cor})$ yields
	\begin{eqnarray*}
		\E\left[\max\{X,Y\}\right] - \E\left[\max\{ X,\widehat Y\}\right] & = & \int_{0}^{\infty} \Pr\left[ X>x\ \mbox{ and } \  \widehat Y > x\right] - \Pr\left[X>x\ \mbox{ and } \  Y > x\right] dx\\
		& = & \int_{0}^{\infty} \Pr\left[X>x \vert \widehat{Y }>x\right]\Pr\left[\widehat Y>x\right] - \Pr\left[ X > x \vert Y>x\right]\Pr\left[Y>x\right] dx\\
		& = & \int_{0}^{\infty} \Pr\left[ Y>x\right]\left(\Pr\left[ X > x\right] - \Pr\left[X>x \vert  Y>x\right]\right) dx\\
		& \leq & 0,
	\end{eqnarray*}
	where the third equality follows from the fact that $\widehat{Y}$ has the same marginals as $Y$ and from the independence of ${X}$ and $\widehat{Y}$, and the inequality follows from the fact that $X$ and $Y$ are positively correlated.
\end{proof}

\begin{lemma}\label{lem:pos_corr}
	For any regular distribution $F$, and $\mathbf{a}$ and $\mathbf{b}$ sampled from $F^\ell$ and $F^k$ respectively, $\varphi\left(a_{(1)}\right)_{\vert a_{(1)} > b_{(1)}}$  and ${a_{(2)}}_{\vert a_{(1)} > b_{(1)}}$ are positively correlated.
\end{lemma}
\begin{proof}
	Fix some value of $x$, and let $v$ be the (highest) value for which $\varphi(v)=x$, where $\varphi$ is the virtual valuation function of $F$.
	We have
	\begin{eqnarray*}
		\Pr_{\substack{\textbf{a}\sim F^{\ell}\\ \textbf{b}\sim F^{k}}}\left[\varphi\left(a_{(1)}\right) > x  \ \vert \ a_{(1)} > b_{(1)}\right] & = & \Pr_{\substack{\textbf{a}\sim F^{\ell}\\ \textbf{b}\sim F^{k}}}\left[a_{(1)} > v \ \vert \  a_{(1)} > b_{(1)}\right]\\
		& \leq & \Pr_{\substack{\textbf{a}\sim F^{\ell}\\ \textbf{b}\sim F^{k}}}\left[a_{(1)} > v \  \vert \ a_{(1)} > b_{(1)} \ \wedge \ a_{(1)} > y \right] \\
		& \leq & \Pr_{\substack{\textbf{a}\sim F^{\ell}\\ \textbf{b}\sim F^{k}}}\left[a_{(1)} > v \ \vert a_{(1)} \  > b_{(1)} \ \wedge \ a_{(2)} > y \right] \\
		& = & \Pr_{\substack{\textbf{a}\sim F^{\ell}\\ \textbf{b}\sim F^{k}}}\left[\varphi\left(a_{(1)}\right) > x \  \vert \ a_{(1)} > b_{(1)} \ \wedge \ a_{(2)} > y \right],
	\end{eqnarray*}
	where the first and last equalities follow by the regularity of $F$, and the last inequality follows from the fact that whenever $a_{(2)} > y$ then also $a_{(1)}> y$. This immediately implies a positive correlation between $\varphi\left(a_{(1)}\right)_{\vert a_{(1)} > b_{(1)}}$ and ${a_{(2)}}_{\vert a_{(1)} > b_{(1)}}$.
\end{proof}

\begin{lemma}\label{lem:dominance}
	For every distribution $F$, $\mathbf{a}$ and $\mathbf{b}$ sampled from $F^\ell$ and $F^k$ and $\mathbf{c}$ sampled from $F^{\ell+k}$, $c_{(2)}$ FOSD ${a_{(2)}}_{\vert a_{(1)} > b_{(1)}}$.
\end{lemma}
\begin{proof}
	Fix some value $x$ in the support of $F$. We need to show that $$\Pr_{\textbf{c}\sim F^{\ell+k}}\left[ c_{(2)}> x\right]\geq \Pr_{\substack{\textbf{a}\sim F^{\ell}\\ \textbf{b}\sim F^{k}}}\left[ a_{(2)}> x \ \vert \  a_{(1)}>b_{(1)}\right].$$
	For a given vector $\textbf{c}$ sampled from $F^{\ell+k}$, let $\textbf{c}^{a}$ be the first $\ell$ coordinates of $\mathbf{c}$,  and $\textbf{c}^b$ be the last $k$ coordinates. We notice that the event that $c_{(1)}$ resides in the first $\ell$ coordinates is independent of the event that $c_{(2)}>x$. Therefore,
	\begin{eqnarray*}
		\Pr_{\textbf{c}\sim F^{\ell+k}}\left[ c_{(2)} > x\right] & = & \Pr_{\textbf{c}\sim F^{\ell+k}}\left[ c_{(2)}> x\vert c^{a}_{(1)}>c^{b}_{(1)}\right] \\
		& \geq &\Pr_{\textbf{c}\sim F^{\ell+k}}\left[ c^a_{(2)} > x\vert c^{a}_{(1)}>c^{b}_{(1)}\right]\\
		& = & \Pr_{\substack{\textbf{a}\sim F^{\ell}\\ \textbf{b}\sim F^{k}}}\left[ a_{(2)}> x \ \vert \  a_{(1)}>b_{(1)}\right],
	\end{eqnarray*}
	where the inequality follows since whenever $c_{(2)}^a>x$ then also $c_{(2)}>x$, and the last equality follows by renaming $\textbf{c}^a$ to $\textbf{a}$ and $\textbf{c}^b$ to $\textbf{b}$.
\end{proof}

Combining the aforementioned three lemmas three lemmas yields the claim.

\noindent\textbf{Proof of Claim \ref{cor:max_fresh_geq_max_same}:} Consider the random variables $X=\varphi\left(a_{(1)}\right)_{\vert a_{(1)} > b_{(1)}}$, $Y={a_{(2)}}_{\vert a_{(1)} > b_{(1)}}$, and $\hat{Y}=c_{(2)}$.
By Lemma~\ref{lem:pos_corr} $X$ and $Y$ are positively correlated, and by Lemma~\ref{lem:dominance} $\hat{Y}$ FOSD $Y$.
In addition, $\hat{Y}$ and $X$ are independent.
The Claim then follows from Lemma~\ref{lem:corr}. \qed



\section{Additive Subject to Matroid Feasibility Constraints}\label{sub:matroids}
In this section we establish an improved version of Lemma \ref{lem:downward-extra} for the special case of matroids. We use the following notation: given a product of regular distributions $F$ and a $n'$ bidders subject to \textit{identical} matroid feasibility constraints, let $\VCG_{n'}^{\Mat}$ denote the expected VCG revenue from selling the $m$ items to $n'$ bidders with i.i.d.~values drawn from $F$, whose valuations are additive subject to the matroid constraint.

Let $([m],\Feas)$ be a matroid set system with rank function $r$. A set $S\subset M$ \emph{spans} an element $j\in M\setminus S$ if $r(S\cup\{j\})=r(S)$ (i.e., the rank of $S$ does not increase when $j$ is added to it).
We introduce the following parameter $\rho$ of a matroid $([m],\Feas)$, which we call the matroid's \emph{disjoint spanning number}:
\begin{definition}[Disjoint spanning number]
	For an element $j$ of the matroid, let a \emph{disjoint spanning collection of $j$} be a family $\Feas'\subset \Feas$ of disjoint independent sets such that every $I\in\Feas'$ spans $j$. 
	Let $\rho_j$ denote the maximum size $|\Feas'|$ over all disjoint spanning collections $\Feas'$ of $j$.\footnote{For example, in a graphical matroid, $\rho_j$ would be the maximum number of cycles that share the edge $j$ and are otherwise disjoint.} We define $\rho:=\max_j\{\rho_j\}$.
\end{definition}
  As an example, we state the value of $\rho$ for two kinds of simple matroids:

\begin{example}[Examples of disjoint spanning numbers]
	\label{obs:rho-simple-matroids}
	For $k$-uniform matroids where $k<m$, $\rho=\lceil(m-1)/k\rceil$. For partition matroids, $\rho$ is the size of the largest partition.
\end{example}

We can now state our main lemma for matroids, which shows the extra competition complexity arising from imposing matroid constraints on additive bidders:

\begin{lemma}[Main lemma for matroids]
	\label{lem:matriod-extra}
	For every $n+\rho$ bidders with a matroid feasibility constraint whose maximum disjoint spanning number is $\rho$,
	$\VCG_{n+\rho}^{\Mat} \ge \VCG_{n}^{\Add}$.
\end{lemma}

The proof of Theorem \ref{cor:matroid-main} is the same as the proof of Theorem \ref{thm:downward-main} only with Lemma \ref{lem:downward-extra} replaced by Lemma \ref{lem:matriod-extra}. Before proving Lemma \ref{lem:matriod-extra}, we briefly discuss its relation to the work of \cite{RTY15} for the special case of unit-demand bidders.

\subsection{Unit-demand Valuations \cite{RTY15}}
\label{sub:unit-demand}

\cite{RTY15} prove a version of Lemma \ref{lem:matriod-extra} for the case of unit-demand bidders. Such bidders can alternatively be described as additive bidders subject to $1$-uniform matroid constraints.
In particular, \cite{RTY15} show that $\VCG_{n+m}^{\UD} \ge \VCG_{n+1}^{\Add}$, which is precisely what we get if we apply Lemma \ref{lem:matriod-extra} to 1-uniform matroids, using that the disjoint spanning number $\rho$ of 1-uniform matroids is equal to $m-1$ (Example \ref{obs:rho-simple-matroids}). In fact this is the same inequality shown in Lemma \ref{lem:downward-extra} for general downward-closed constraints. We conclude that the case of unit-demand bidders is the ``hardest'' one in terms of competition complexity, not only among matroids but also among all downward-closed set systems.

We generalize the result of \cite{RTY15} in two ways: first, by extending it to all downward-closed constraints (Lemma \ref{lem:downward-extra}); second, by introducing a parameterized version for all matroids using the environment disjoint spanning number as the parameter (Lemma \ref{lem:matriod-extra}). In terms of techniques, our result for matroids uses a stability argument similar to that in \cite{RTY15}. Our result for general downward-closed uses a different charging argument as described in Section~\ref{sub:main-lemma-pf}. Both results use the principle of deferred decision as used in the analysis of \cite{RTY15} in order to argue that the values of unallocated bidders (which determine VCG payments) are sufficiently high in expectation.

\subsection{Proof of Main Lemma (\ref{lem:main_additive}) for Matroids}
\label{sub:matroid-lemma-pf}

The next claim will be used in the proof of Lemma \ref{lem:matriod-extra}. It follows easily from the definitions of matroid rank function and spanning in matroids.

\begin{claim}
	\label{cla:span}
	Consider a matroid $([m],\Feas)$. For every subset of elements $S\subseteq [m]$, and for every two elements $j,j'\in [m]\setminus S$ such that $S$ spans $j$ but does not span $j'$, $S\cup\{j\}$ does not span $j'$. Moreover, for every subset $S'\subseteq S\cup\{j\}$, $S'$ does not span $j'$.
\end{claim}

We state the following claim for all gross substitutes valuations, and apply it in the proof of Lemma \ref{lem:matriod-extra} to the subclass of additive valuations subject to matroid constraints. The claim uses the language of \emph{vacated} and \emph{snatched} items as defined in the proof of Lemma \ref{lem:downward-extra}.

\begin{claim}
	\label{cla:GS}
	Consider a welfare-maximizing allocation $\OPT_{-j}$ of the $m-1$ items in $M\setminus\{j\}$ to $n$ bidders with gross substitutes valuations. A welfare-maximizing allocation $\OPT$ of the $m$ items in $M$ to the $n$ bidders can be found by a chain of reallocations starting from $\OPT_{-j}$, where in each step a bidder snatches either item $j$ (in the first step) or a single vacated item (in subsequent steps), and possibly vacates another single item.
\end{claim}

\begin{proof} [Sketch]
	Let $p_{-j}$ be a vector of Walrasian prices \emph{supporting} the welfare-maximizing allocation $\OPT_{-j}$ \cite{GS99}, i.e., each bidder is allocated his \emph{demand} (utility-maximizing bundle of items) at prices $p_{-j}$, and no unallocated item is priced above 0. Now add item $j$ with an initial price $p_j$ that is prohibitively high. We show a descending-price process that transforms $\OPT_{-j}$ into $\OPT$. This process establishes the existence of a chain of reallocations as claimed above (its computational properties are irrelevant in our context).
	
	Start by lowering $p_j$ until there is a single%
	\footnote{We are assuming no ties here to simplify the exposition.} %
	bidder $i$ who is no longer getting his demand or until $p_j$ reaches $0$.%
	\footnote{In full formality, $p_j$ should be lowered by $\epsilon$-increments and $\epsilon$ should be taken to 0.}
	If $p_j=0$ we have found $\OPT$ (to see this, observe that since each bidder is allocated his demand given current prices and no item with positive price is unallocated, we have reached a Walrasian equilibrium whose allocation is welfare-maximizing by the first welfare theorem \cite{GS99}). Otherwise, by the \emph{single improvement} property of gross substitutes \cite{GS99}, and since prices $p_{-j}$ \emph{support} the allocation of items $M\setminus{j}$, bidder $i$ can improve his utility by and only by snatching the single item whose price has been reduced, which in the first step is item $j$, and then possibly vacating another single item. If no item is vacated, we have again reached a welfare-maximizing equilibrium allocation $\OPT$. Otherwise, if an item $j'$ was vacated, then we repeat the above process of lowering the price of item $j'$. This process must end with a welfare-maximizing allocation $\OPT$, thus completing the proof sketch.
\end{proof}

We now prove our main lemma for matroids.

~

\noindent\textbf{Proof of Lemma \ref{lem:matriod-extra}:}
Our goal is to show that $\VCG_{n+\rho}^{\Mat} \ge \VCG_{n}^{\Add}$. 
As in the proof of Lemma~\ref{lem:downward-extra}, it is sufficient to identify $n$ out of the $n+\rho$ constrained bidders whose values for item $j$ are i.i.d.~draws from $F_j$, such that the VCG payment for $j$ is at least as high as the second highest value of these bidders for $j$.
This establishes that the contribution from selling item $j$ to $\VCG_{n+\rho}^{\Mat}$ is at least $\E[u_{(2:n)j}]$, which is exactly the contribution from selling item $j$ to $\VCG_{n}^{\Add}$.

As in the proof of Lemma \ref{lem:downward-extra} we fix an item $j$, and fix the values of the $n$ constrained bidders for all items but~$j$. We denote the welfare-maximizing allocation of all items but $j$ to the $n$ constrained bidders by $\OPT_{-j}$. Diverging from the proof of Lemma \ref{lem:downward-extra}, we denote by $A$ the set of bidders \emph{whose allocation in $\OPT_{-j}$ spans $j$}. Notice that $|A|\le \rho_j\le\rho$, and so $|\bar{A}|\ge n$. We further fix the values of the bidders in $A$ for item $j$, so that all remaining randomness is in the values of bidders in $\bar{A}$ for item $j$, which are i.i.d.~random samples from $F_j$. We denote by $\OPT$ the random welfare-maximizing allocation of all items (including $j$) to the $n$ bidders.

Assume first that item $j$ is allocated in $\OPT$ to a bidder $i$ in $\bar{A}$. As in the proof of Lemma \ref{lem:downward-extra}, the allocation of the other items in $\OPT_{-j}$ does not change in $\OPT$.
Since all bidders in $\bar{A}$ but $i$ still have allocations that do not span item $j$, the VCG payment for $j$ is at least as high as the second-highest value of the bidders in $\bar{A}$ for $j$.

Assume now that item $j$ is allocated in $\OPT$ to a bidder in $A$. Note that as in the proof of Lemma \ref{lem:downward-extra}, we can find $\OPT$ without knowing ${\bf v}_{\bar{A}j}$ (if there are several welfare-maximizing allocations, we assume that VCG finds the one implied by Claim \ref{cla:GS}). By Claim \ref{cla:GS} we know that the bidders in $\bar{A}$ have changed their allocation only by ``single improvements'' involving snatching an item and possibly vacating another.
In fact, in all such single improvement but one, an item was vacated.
Since bidders' valuations are additive subject to matroid constraints, a bidder $i$ with allocation $S_i$ vacates an item $j'$ if and only if he first snatches an item $j''$ such that $S_i$ spans $j''$. By Claim \ref{cla:span}, this means that the eventual allocation of bidder $i$ in $\OPT$ still does not span item $j$. We conclude that the VCG payment for $j$ is at least as high as the second-highest value of the bidders in $\bar{A}$ for $j$.

We get that in both cases, the payment for $j$ is at least the second-highest value of the bidders in $\bar{A}$. Since the set $\bar{A}$ is determined before sampling the values for item $j$, this completes the proof of the lemma.
\qed

\section{Additive Subject to Asymmetric Downward Closed Feasibility Constraints}\label{sub:asymmetric}
In this section, we consider the case of heterogeneous downward-closed feasibility set systems for the bidders. 
Let $([m],\mathcal{I}_i)$ denote the feasibility set system of bidder $i$. 
For our results to hold, we add the mild assumption that $\{j\}\in \mathcal{I}_i$ for every bidder $i$ and item $j$.
 We first point out the part where the proof for symmetric feasibility constraints fails. Unlike the case where bidders have symmetric feasibility constraints, the payment of an bidder is not lower bounded by the sum of the  highest values of the unallocated bidders to the items; this is in contrast to Lemma \ref{lem:vcg_rev}\footnote{To see this, consider an additive bidder who is allocated two items, $a$ and $b$, by the $\VCG$ mechanism, and two unit demand bidders, $1$ and $2$, where $v_1(a)=3$, $v_1(b)=5$, $v_2(a)=1$, $v_2(b)=4$. The payment of the additive bidder is $3+4=7<3+5=8$.}. Therefore, we need a new lower bound on the revenue of $\VCG$. Before stating the new lower bound, we introduce the following notation. We order the items in lexicographic order, and for a given set $S$, we define:
\squishlist
	\item $1^*(S)=\arg\max_{i\in S} v_i(1)$ (the bidder in set $S$ with the highest value for  item $1$).
	\item $S_{>j}=  S\setminus\left\{1^*(S),2^*(S),\ldots, j^*(S)\right\}$, contains all bidders except for $1^*(S),2^*(S),\ldots, j^*(S)$.
	\item $j^*(S)=\arg\max_{i\in S_{>j-1}} v_{i}(j)$. 
	\item $r_j(S)=v_{j^*(S)}(j)$ (the value of item $j$ to bidder $j^*(S)$).
\squishend
We give a lower bound on the revenue of $\VCG$ given the above definitions.
\begin{lemma} \label{lem:vcg_rev_asym}
	Consider a set of bidders with asymmetric downward-closed feasibility constraint, and the outcome of $\VCG$. Let $\bar{A}$ be the set of unallocated bidders. Assuming the number of bidders is at least $2m$, the revenue of $\VCG$ is at least $\sum_{j\in [m]} r_j\left(\bar{A}\right)$.
\end{lemma}
\begin{proof}
	Let $A$ and $\bar{A}=[n]\setminus A$ be the respective sets of allocated and unallocated bidders. Let $S=(S_1,\ldots, S_n)$ be the allocation returned by $\VCG$, and let $W$ be the welfare of this allocation. Fix an agent $i$ that received an item. The welfare of all agents but $i$ for this outcome is $W-v_i(S_i)$. Consider now the following allocation --- allocate the items not in $S_i$ as in $S$, and distribute each item $j\in S_i$ to an unallocated bidder $j^*\left(\bar{A}\right)$ (notice that this is well defined since the number of agents is assumed to be at least $2m$). By construction, $j^*(\bar{A})\neq {j'}^*(\bar{A})$ for every $j\neq j'$. Therefore, each bidder in $\bar{A}$ gets at most a single item, which is feasible under our assumption that $\{j\}\in\mathcal{I}_i$ for every $i,j$.
	The welfare of the proposed allocation is $W-v_i(S_i)+\sum_{j\in S_i}r_j(\bar{A})$. Since we described some allocation of items to all bidders but $i$, the welfare of the optimal allocation to bidders without $i$ is at least as high. In $\VCG$, each agent pays her ``externality on the other agents" (\textit{i.e.}, the loss in welfare of the other agents due to her existence). We get that the payment of bidder $i$ is at least $\sum_{j\in S_i}r_j(\bar{A})$. Summing over all bidders completes the proof of the lemma.
\end{proof}

Similarly to the proof of Theorem \ref{thm:downward-main} and Corollary \ref{cor:matroid-main}, we prove a lemma that shows how many extra bidders are needed for the revenue of $\VCG$ with arbitrary downward closed constrained bidders to exceed the revenue of $\VCG$ with unconstrained bidders. To this end we fix a product distribution $F$ of $m$ regular distributions, and denote by $\VCG_{n'}^{\Asym}$ the expected $\VCG$ revenue from selling $m$ items to $n'\geq n$ bidders with i.i.d.~values drawn from $F$, whose valuations are additive subject to $n'$ (not-necessarily identical) downward-closed feasibility constraints.
\begin{lemma}
	\label{lem:assym-extra}
	For every $n+2(m-1)$ bidders with arbitrary downward closed feasibility constraint,
	$\VCG_{n+2m-2}^{\Asym} \ge \VCG_{n}^{\Add}$.
\end{lemma}

Lemma~\ref{lem:assym-extra} is useful in its own right, but also can be used to imply a competition complexity result extending Theorem~\ref{thm:downward-main} to the case where feasibility constraints are not fixed for the entire population.

We now prove Lemma~\ref{lem:assym-extra}, and begin by showing a useful property of the notations used in order to bound $\VCG$'s revenue. 

\begin{lemma}\label{lem:payment_prop}
	For any two sets of of bidders $S$ and $T$ such that $S\subseteq T$, and for every $j\in [m]$, $S_{>j-1}\subseteq T_{>j-1}$.
\end{lemma}
\begin{proof}
	The proof is by induction on $j$. For $j=1$, $S_{>0}=S\subseteq T = T_{>0}$. Now suppose the claim is true for some $j$ \big($S_{>j-1}\subseteq T_{>j-1}$\big). Consider two cases: if $j^*(S)  =j^*(T)$ (denote this bidder by $j^*$), then, by the induction hypothesis, $$S_{>j}= S_{>j-1}- \{j^*\}\subseteq T_{>j-1}- \{j^*\} = T_{>j};$$ otherwise, if $j^*(S)\neq j^*(T)$, then by the induction hypothesis, $j^*(S)\in T_{>j-1}$, thus $j^*(T) \in T_{>j-1}\setminus S_{>j-1}$. Therefore, $$T_{>j}=T_{>j-1}-\{j^*(T)\}\supseteq S_{>j-1}\supset S_{>j}.$$
\end{proof}

As a corollary, we get:
\begin{corollary} \label{cor:payment_mon}
	For any two sets of of bidders $S$ and $T$ such that $S\subseteq T$, and for every $j\in [m]$, $r_j(S)\leq r_j(T)$.
\end{corollary}
\begin{proof}
	From the above lemma, $j^*(S)\in T_{>j-1}$. Therefore, by the definition of $r_j$, it holds that $r_j(T)\geq r_j(S)$.
\end{proof}

We now prove the main lemma for bidders with asymmetric feasibility constrains.

~

\noindent\textbf{Proof of Lemma \ref{lem:assym-extra}:}
Let $\bar{S}$ be the set of bidders unallocated by $\VCG$ on the $n+2(m-2)$ constrained bidders. We prove that for every $j$, the expected value of $r_j(\bar{S})$ is at least as high as the expectation of the second highest value of $n$ draws from $F_j$. This argument suffices to prove the lemma since the latter is exactly what is paid for item $j$ in  $\VCG$ with unconstrained bidders.

As in previous proofs, the proof uses again the principle of deferred decisions. Consider first sampling values of all items but $j$, and then sample the value of item $j$ for all bidders that are allocated some item in the optimal allocation that does not allocate $j$. Similarly to the proof of Lemma \ref{lem:downward-extra}, we can identify a set $\tilde{A}$ of $n+m-1$ bidders out of which at most a single bidder is going to be allocated by $\VCG$. Furthermore, the value of item $j$ for these agents is yet to be determined. We denote the random set that will be left unallocated after putting item $j$ back in by $\tilde{A}^{+j}$.
Consider the set of bidders $\tilde{A}_{< j}= \{\ell^*(\tilde{A}): \ell < j\}$. We claim that no matter which bidder will be allocated from $\tilde{A}$ in the final allocation (if any), all bidders remaining from  $\tilde{A}_{< j}$ will be in $\tilde{A}^{+j}_{<j} = \{\ell^*(\tilde{A}^{+j}): \ell < j\}$. Otherwise, there would be an bidder in $\tilde{A}^{+j}_{>j-1}$ which is not in $\tilde{A}_{>j-1}$, contradicting Lemma \ref{lem:payment_prop}.

Notice now that the set $\tilde{A}_{< j}$ is of size at most $m-1$, and therefore, the set $\tilde{A}_{>j-1}$ is of size at least $n+m-1-(m-1)=n$. Furthermore, the set $\tilde{A}_{>j-1}$ is determined before sampling the values of item $j$ for these bidders.
Since at most a single bidder from $\tilde{A}_{>j-1}$ is in $\tilde{A}^{+j}_{<j}$, $j^*(\tilde{A}^{+j})$ is one of the two bidders with the highest value for item $j$ in $\tilde{A}_{>j-1}$. Therefore, by Corollary \ref{cor:payment_mon}, $r_j(\bar{S})\geq r_j(\tilde{A}^{+j})$, which is at least the payment for item $j$ in $\VCG$ with $n$ unconstrained bidders. Applying Lemma \ref{lem:vcg_rev_asym} gives us that the revenue of $\VCG$ with $n+2(m-1)$ constrained bidders is at least the  revenue of $\VCG$ with $n$ constrained ones, as desired.\qed

\noindent\textbf{Proof of Theorem~\ref{thm:ccasymmetric}:}
First, draw the constraints independently for each bidder. As constraints are drawn independently of values, we can apply Lemma~\ref{lem:assym-extra}, and conclude that for all constraints that could possibly have been drawn, the expected revenue of $\VCG$ for $n+2m-2$ bidders subject to these (asymmetric) constraints is at least the expected revenue of $n$ additive bidders, conditioned on the drawn constraints. Therefore, when taking an expectation over the randomly drawn constraints, the theorem statement holds.
\qed

\subsection{Competition Complexity of Maximum in Range $\VCG$} \label{app:mir_vcg}
We furthermore show how to overcome the intractability of $\VCG$ for general downward-closed feasibility constraints, even when the constraints are asymmetric across agents. Recall that $\VCGUD$ is defined to run $\VCG$, and restrict the outcomes to matchings -- \textit{i.e.}, every agent gets at most a single item. This is a maximal-in-range mechanism, and hence, truthful. Furthermore, since finding maximum weight matchings is tractable, this is also a poly-time prior independent mechanism. We can now ask the following question --- what is the the competition complexity of this mechanism (if well defined). 

We first provide a lower bound on the revenue of $\VCGUD$. 
\begin{observation}
	Consider a set of agents with asymmetric downward-closed feasibility constraint, and the outcome of $\VCGUD$. For each item $j$, let $v^*_j$ be the highest value of an unallocated agent for item $j$. The revenue of $\VCGUD$ is at least $\sum_{j\in [m]} v^*_j$.
\end{observation}
\begin{proof}
	Similarly to the proof  of Lemma \ref{lem:vcg_rev}, this follows from taking the optimal allocation, removing the agent that is allocated item $j$, and allocating this item to the unallocated agent with the highest value for this item.
\end{proof}

Notice that we can use this observation, and an essentially identical proof to the proof of Lemmas \ref{lem:downward-extra} and \ref{lem:matriod-extra} (and similar to the claims used in \cite{RTY15}), to show the following:

\begin{lemma}
	\label{lem:vcgmir-extra}
	For every $n+m-1$ bidders with a asymmetric downward closed feasibility constraints,
	$\VCGUD_{n+m-1} \ge \VCG_{n}^{\Add}$.
\end{lemma}

This in turn yields the following upper bound on the competition complexity of $\VCGUD$, that is, the number of extra bidders one needs to add to the original environment in order the get revenue greater than the optimal revenue of the original  environment using mechanism $\VCGUD$, and a complete proof of Theorem~\ref{thm:vcgud}.

\section{Lower Bounds on the Competition Complexity} \label{sec:lb}

\begin{theorem}
\label{thm:lower-bound-additive}
The competition complexity of a single bidder with additive valuations over $m$ independent regular items is $\Omega(\log m)$.
\end{theorem}

\begin{proof}
Consider a single additive bidder with $m$ items i.i.d. distributed according to the equal-revenue distribution with CDF $F(x)=1-1/x$.  In Lemma 8 in \cite{HartN12}, it was shown that the optimal revenue for this setting satisfies $\Rev(F^m)\geq m\cdot \log m/4$. 

Now consider a setting with a single item and $k$ bidders, all distributed according to $F$. 
The expected revenue of the optimal single-bidder auction is $1$.
By revenue submodularity \cite{DughmiRS12}, the expected revenue of running the optimal mechanism is a concave function of the number of bidders. Thus, the expected revenue of the optimal auction with $k$ bidders is at most $k$ (and it can actually be shown to be exactly $k$).
Consequently, the expected revenue of running VCG in the single item, $k$ bidder setting is upper bounded by $k$.

Consider now the original setting with a single additive bidder and $m$ items i.i.d. distributed according to the equal-revenue distribution. 
Based on the analysis above, and the linearity of VCG with additive bidders, the expected revenue of running VCG with $k-1$ additional bidders (i.e., with $k$ bidders) is at most $m\cdot k$. 
Thus, in order to surpass the optimal revenue $\Rev(F^m)$, the number of additional bidders ($k-1$) should be at least as large as $\log m/4-1=\Omega(\log m)$. 
This establishes the assertion of the theorem.
\end{proof}

\section{$1/2$ approximation of $\VCG$ as $n$ goes to $\infty$} \label{app:large_market}
Here we show that as a direct corollary of our results, running $\VCG$ in large markets yields a $2$-approximation on  the optimal revenue for additive bidders.

\begin{theorem}
	For any distribution of regular items $F$, whenever $n\gg m$, then running $\VCG$ on $n$ bidders (without recruiting additional bidders) yields a revenue $\approx$ half of the optimal revenue.
\end{theorem}
\begin{proof}
	Consider a market with $n$ additive bidders with values sampled from $F$. According to Theorems \ref{thm:additive_main}, $\VCG_{2n+m-1}^{\Add}  \ge \Rev_{n}^{\Add}$. Since $\VCG$ decomposes across items for additive bidders, and since $n>m$, revenue submodularity applies \cite{DughmiRS12}. Therefore, $\VCG_{n}^{\Add}\geq \frac{n-m}{2n+m-1-m}\VCG_{2n+m-1}^{\Add}\geq \frac{n-m}{2n-1}\Rev_{n}^{\Add}\approx_{n\rightarrow\infty}\Rev_{n}^{\Add}/2$.
\end{proof}

\end{document}